\documentclass{article}

\usepackage{PRIMEarxiv}

\usepackage[utf8]{inputenc} 
\usepackage[T1]{fontenc}    
\usepackage{hyperref}       
\usepackage{url}            
\usepackage{booktabs}       
\usepackage{amsfonts}       
\usepackage{nicefrac}       
\usepackage{microtype}      
\usepackage{lipsum}
\usepackage{fancyhdr}       
\usepackage{graphicx}       
\usepackage{xcolor}         
\graphicspath{{media/}}     

\usepackage{caption}
\usepackage{subcaption}

\usepackage{amsmath}
\usepackage{amsthm}
\usepackage{mathtools}
\usepackage{natbib}
\usepackage{dsfont}

\newtheorem{proposition}{Proposition}

\newtheorem{definition}{Definition}
\newtheorem{algo}{Algorithm}
\newtheorem{lemma}{Lemma}

\providecommand{\R}{}
\renewcommand{\R}{\mathds{R}}

\newcommand\cL{{\mathcal L}}
\newcommand\cM{{\mathcal M}}
\newcommand\cN{{\mathcal N}}
\newcommand\cO{{\mathcal O}}
\newcommand\cS{{\mathcal S}}

\newcommand{\ind}{\mathds{1}}

\renewcommand{\vec}{\mathrm{vec}}
\renewcommand{\det}{\mathrm{det}}
\newcommand{\diag}{\mathrm{diag}}
\newcommand{\trace}{\mathrm{tr}}
\newcommand{\etr}{\mathrm{etr}}

\newcommand{\new}{\mathrm{PRED}}
\newcommand{\old}{\mathrm{OBS}}

\newcommand{\var}{\mathrm{Var}}
\newcommand{\cov}{\mathrm{Cov}}

\newcommand{\mat}[1]{\boldsymbol{\mathrm{#1}}}

\newcommand{\LMC}{{\text{LMC}}}

\newcommand{\indep}{\ensuremath{\stackrel{\text{ind.}}{\sim}}}

\title{Computational Considerations for the Linear Model of Coregionalization
}

\author{
  Renaud Alie, David A. Stephens \\
  Department of Mathematics and Statistics \\
  McGill University \\
  Montreal\\
   \And
  Alexandra M. Schmidt \\
  Department of Epidemiology, Biostatistics and Occupational Health\\
  McGill University \\
  Montreal\\
}

\begin{document}

\maketitle

\begin{abstract}
In the last two decades, the linear model of coregionalization (LMC) has been widely used to model multivariate spatial processes. However, it can be a challenging task to conduct likelihood-based inference for such models because of the cubic cost associated with Gaussian likelihood evaluations. Starting from an analogy with matrix normal models, we propose a reformulation of the LMC likelihood that highlights the linear, rather than cubic, computational complexity as a function of the dimension of the response vector. We describe how those simplifications can be exploited in Gaussian hierarchical models. In addition, we propose a new sparsity-inducing approach to the LMC that introduces structural zeros in the coregionalization matrix in an attempt to reduce the number of parameters in a principled and data-driven way. Our reformulation of the LMC likelihood ensures that our sparse approach comes at virtually no additional cost when included in a Markov chain Monte Carlo (MCMC) algorithm. It is shown, on synthetic data, to significantly improve predictive performance. We also apply our methodology to a dataset comprised of air pollutant measurements from the state of California. We investigate the strength of the correlation among the measurements by providing new insights from our sparse method.

\end{abstract}

\section{Introduction}

The literature concerning computational scalability for spatial Gaussian models has for the most part been concerned with applications where the number of spatial locations $n$ is large. Such situations usually forbid likelihood-based inference because of the $\cO(n^3)$ computational cost associated with evaluating the joint normal density. Some of the most prominent approximations used in the applied literature include predictive processes \citep{banerjee2008gaussian,finley2009improving}, the stochastic partial differential equation approach pioneered in \citet{lindgren2011explicit} or the nearest-neighbor methodology \citep{vecchia1988estimation,datta2016hierarchical}. \citet{heaton2019case} provide an extensive comparative study of those methods and more. In multivariate spatial models where an output of dimension $p$ is modeled, the cost associated with Gaussian process likelihood evaluations is, in general, $\cO(p^3n^3)$. While the approximations mentioned above are expressed in multivariate generality, likelihood computations still scale with the number $p$ of dimensions cubed.

Multivariate spatial models necessitate the specification of a cross-covariance function. \citet{genton2015cross} describe the necessary conditions under which such a function will lead to a valid model. They also review some of the most widespread approaches. A lot of those fall under the category of constructive methods in which cross-covariance functions are built from univariate models. Those include the covariance convolution framework \citep{majumdar2007multivariate} or the approach based on latent dimensions \citep{apanasovich2010cross}, among others. Another important example is the class of Mat\'ern cross-covariance functions \citep{gneiting2010matern,apanasovich2012valid} which offers, under some constraints, multivariate spatial models where the marginal distribution of each process has univariate Mat\'ern covariance structure. None of the models mentioned above benefit from the type of computational simplifications we describe in this paper.

Perhaps the simplest (at least conceptually) of multivariate models, the linear model of coregionalization (LMC), is constructed as a linear combination of independent processes.
From $r$ orthogonal base processes, we obtain $p$ new components that have cross-dependence. In its inception \citep{matheron1982pour,bourgault1991multivariable,goulard1992linear}, the LMC was conceptualized as a mechanism allowing for dimensionality reduction. This is accomplished by specifying a transformation matrix of dimension $p \times r$ with $r < p$. It essentially shrinks an observed $p$-dimensional process to $r$ latent and independent ones. In this work, we only discuss the $r = p$ case. We exclusively consider the LMC as an approach to construct a $p$-variate process with cross-dependence from $p$ univariate spatial correlation structures as it is described in \cite{schmidt2003bayesian,gelfand2004nonstationary,gelfand2005spatial}. The LMC is still used to this day in the applied sciences (see for example \citet{ji2021geographically, carter2024land, heydari2023scale}). It is moreover an excellent candidate to incorporate into more involved methodology, such as spatiotemporal models \citep{mastrantonio2019hierarchical,cappello2022computational} or neural network architectures \citep{liu2022scalable,wang2022lmc}.

The LMC implies a likelihood that can be evaluated in a number of operations that is linear in $p$. This offers immense speed-ups in cases where the number of dimensions is even moderately large. This is not the result of an approximation. Instead, we achieve those computational shortcuts by exploiting the cross-covariance structure implied by the LMC. Hence, there are no drawbacks to utilizing the computational results we present in this paper. Those computational shortcuts can be viewed as a natural extension of separable models. 

Work on so-called highly multivariate processes concerns spatial statistical applications where the number $p$ of dimensions and the number $n$ of locations are both fairly large, with the former being at least a significant fraction of the latter. \citet{taylor2019spatial} and \citet{zhang2022spatial} describe similar approaches where a spatial factor model is compounded with a nearest-neighbor approximation to handle high-dimensional processes observed at numerous locations. In the context of areal data, \citet{bradley2015multivariate} proposed to combine a dimension reduction mechanism with a specific basis function expansion of the spatial process, creating a model that handles both large $p$ and $n$. In this work, we focus on the advantages of working with the LMC as a multivariate model. The computational simplifications associated with this model could be implemented in any of the frameworks described above as an alternative to dimension reduction.

Other approaches to highly multivariate data reduce the computations in terms of $p$ by exploiting a sparser dependence structure among processes. \citet{dey2022graphical} propose a stitching construction that conforms to a conditional independence graph among processes while preserving the specified marginal distribution of each one. Sparse and decomposable graphs lead to easier evaluated likelihood functions with fewer parameters. Whenever the conditional independence graph is not known in advance, the authors propose a reversible jump MCMC algorithm on the edges of the conditional independence graph. In \cite{krock2023modeling}, the authors assume a basis expansion of the spatial Gaussian process. They obtain a sparse dependence structure by formulating the estimation of the covariance among basis coefficients as a graphical lasso optimization problem. In contrast, we show that the LMC, as a multivariate model, implies likelihood computations that scale linearly in terms of the number $p$ of components even when allowing complete dependence among processes. Additionally, we propose a sparse version that aims to reduce the size of the parameter space associated with the regular LMC. As a by-product, it also allows us to infer sparser dependence structures without increasing the computing time when compared to the regular LMC.

We first review the equivalence between separable models and matrix normal distributions in Section \ref{LMCmain}. We then extend this matrix normal analogy to the non-separable LMC. In the context of MCMC-based Bayesian inference, the structure of the LMC can be advantageous beyond the obvious computational speed-ups. We introduce an alternative prior to the coregionalization matrix in Section \ref{sparseApp}. This introduces sparsity both in the parameter space and the resulting marginal covariance of the $p$ components. Because of our matrix normal formulation of the LMC likelihood, this new approach comes at virtually no additional cost from a compuational standpoint. It is also shown to significantly improve the accuracy of out-of-sample predictions.

\section{Linear Model of Coregionalization}\label{LMCmain}

\subsection{Separable Specification}\label{separable}

In this section, we explore the link between multivariate Gaussian processes with a separable cross-covariance function and the matrix normal distribution. This will serve as an analogy for the computational results we present in Section \ref{genLMC} concerning the non-separable LMC.

\begin{definition}[Matrix Normal Distribution \citep{dawid1981some}]
    A random $p \times n$ matrix $\mat X$ has the $\cM \cN (\mat M, \mat \Sigma,\mat R)$ distribution if its probability density function is given by
    \begin{align}
    p(\mat X | \mat M, \mat \Sigma,\mat R) = \frac{\exp [ -\frac{1}{2} \trace  \{ \mat R^{-1} (\mat X - \mat M)^\top \mat \Sigma^{-1}  (\mat X - \mat M) \} ]}{(2\pi)^{pn/2}  \det(\mat R)^{p/2} \det(\mat\Sigma)^{n/2} }, \label{matNorm}
    \end{align}
    where $\mat M$ is a real valued $p \times n$ matrix and $\mat \Sigma,\mat R$ are respectively $p \times p$ and $n \times n$ positive definite matrices. 
    
    Equivalently, $\mat X \sim \cM \cN (\mat M, \mat \Sigma,\mat R)$ if the vector $\vec(\mat X)$ obtained by stacking the columns $\mat x_1, \dots, \mat x_n$ of $\mat X$ has the multivariate normal distribution with mean $\vec(\mat M)$ and covariance $ \mat R \otimes \mat \Sigma$ where $\otimes$ represents the Kronecker product.
\end{definition}

In the previous definition, $\mat M$ is the location parameter while $\mat \Sigma$ (resp. $\mat R$) corresponds to the marginal column (resp. row) covariance. A realization of the $\cM \cN (\mat M, \mat \Sigma,\mat R)$ distribution can be obtained from a $p \times n$ matrix $\mat Z$ of independent standard normal by setting $\mat X=\mat A \mat Z \mat B^\top + \mat M$ where $\mat A, \mat B$ are such that $\mat \Sigma = \mat A \mat A^\top$ and $\mat R = \mat B \mat B^\top$ ($\mat A$ and $\mat B$ can be the lower triangular Cholesky factors of $\mat \Sigma$ and $\mat R$ for example). Importantly, simulation from such a matrix normal distribution can be accomplished without factorizing the full $np \times np$ matrix $\mat R \otimes \mat \Sigma$. The structure of the model allows us to work with individual matrices $\mat R$ and $\mat \Sigma$. This will be a common thread in the remainder of this article.

In the spatial modeling context, the equivalent would be to define a multivariate process as $\mat v(\mat s) = \mat A \mat w(\mat s)$ where $\mat A$ is a $p \times p$ full rank matrix and $\mat w(\mat s) = (w_1(\mat s),\dots,w_p(\mat s))^\top$ consists of $p$ i.i.d. Gaussian random fields. This is known as the LMC with separable covariance structure \citep{gelfand2004nonstationary}. We assume that each $w_j(\cdot), j=1,\dots,p$ has 0 mean and a marginal variance of 1. Scaling is handled by $\mat A$ which is referred to as the coregionalization matrix in the following. Let $\mat V = (\mat v(\mat s_1),\dots,\mat v(\mat s_n))$ denote a $p\times n$ matrix organizing in columns some $n$ observations of this multivariate process at distinct locations $\mat s_1, \dots, \mat s_n$ contained in some domain $\cS$. Commonly, we consider $\cS \subset \R^2$ such as in the case of geographically indexed observations. Under the separable specification, the distribution of $\mat V$ is $\cM \cN (\mat 0, \mat A \mat A^\top,\mat R)$ where $\mat R$ is the common correlation matrix associated with the $w_j(\cdot)$ processes at locations $\mat s_1.\dots,\mat s_n$. It depends on the coregionalization matrix only through the product $\mat A \mat A^\top$. Such a process admits a cross-covariance function of the form $C(\mat s,\mat t) = \rho(\mat s,\mat t)\mat A \mat A ^T $, where $\rho(\cdot,\cdot)$ is the common correlation function of processes $w_j(\cdot), j=1,\dots,p$. The term separable stems from the fact that the covariance structure factorizes in a spatial component times the marginal (at a single location) covariance among the $p$ dimensions. In this case, inference can be conducted directly on the positive definite matrix $\mat \Sigma = \mat A \mat A^\top$ rather than the non-identifiable $\mat A$. The remaining parameters are those contained in the common spatial correlation function defining the processes $ w_j(\cdot),j=1,\dots,p$. We assume the correlation between two locations $\mat s, \mat t$ to be of the form $\rho(\varphi |\mat s - \mat t|)$ where $\rho(\cdot)$ is some positive definite function. For ease of exposition, the spatial dependence is encapsulated by a single range parameter $\varphi>0$. More elaborate families could also be envisioned. 

We focus on a Bayesian formulation of inference and consider an MCMC approach that updates parameters $\mat \Sigma$ and $\varphi$ in turn at each iteration. This usually entails evaluating the likelihood at each step of the Markov chain which is known to be costly for multivariate Gaussian densities even with a moderate sample size. If we were to naively consider the multivariate normal distribution of $\vec(\mat V)$, evaluating the likelihood would require computing the inverse and determinant of the covariance matrix $\mat R \otimes \mat \Sigma$. This computation has overall $\cO (p^3n^3)$ asymptotic complexity. However, basic algebra results tell us that $(\mat R \otimes \mat \Sigma)^{-1} = \mat R^{-1} \otimes \mat \Sigma^{-1}$ and $\det(\mat R \otimes \mat \Sigma) = \det(\mat R)^p \det(\mat \Sigma)^n$. Exploiting the structure of this model allows us to conduct our calculation directly on the individual matrices $\mat R$ and $\mat \Sigma$ which is a great advantage over a general multivariate model for a sample size of $np$. An update to the single spatial range parameter $\varphi$ requires us to compute the inverse and determinant of the $n \times n$ matrix $\mat R$ which is the crux of the computation since we consider the number $n$ of samples to greatly exceed the dimension $p$ of each observation, i.e. $p \ll n$.

Moreover, we can rewrite the quadratic product $\vec(\mat V)^\top(\mat R^{-1} \otimes \mat \Sigma^{-1}) \vec(\mat V)$ as a trace to obtain an expression similar to equation \eqref{matNorm}, thus avoiding the need to ever construct or store any $np \times np$ matrix. What is also worth mentioning is that we obtain full-conditional conjugacy by assigning an inverse-Wishart prior to $\mat \Sigma$. Those computational considerations about the separable model are exposed in \citet{banerjee2002prediction} among others. We solely review them here to provide some context for the more general results we present in Section \ref{genLMC}. However, from a modeling perspective, the separable specification is too restrictive for most applications as the components of $\mat v (\cdot)$ have the same marginal distribution. In particular, it would imply that all the variables in the model have the same practical range. This assumption is pretty strong and rather unreasonable when modeling the spatial covariation of different environmental variables, for example.

\subsection{Non-Separable LMC} \label{genLMC}

Allowing the independent base processes $w_j(\cdot),j=1,\dots,p$ to have distinct covariance functions solves the issue of having identically distributed components $v_j(\cdot)$. In light of Section \ref{separable}, a natural question that arises is how much, if any, of the computational advantages can be retained from the separable case. Some of the computational results we employ were discussed in the Ph.D. thesis of \citet{kyzyurova2017uncertainty}. To summarize, the covariance structure of the LMC implies Gaussian likelihood evaluations that scale linearly in $p$ rather than the $\cO(p^3n^3)$ associated with general multivariate models. It positions the LMC as a computationally advantageous alternative to other multivariate processes. In our opinion, this fact has failed to gain the traction it deserves.

The setup is as before, with a multivariate process $\mat v(\mat s)$ defined from $p$ independent processes $ w_j(\mat s),j=1,\dots,p$ through the linear transformation $\mat v(\mat s) = \mat A \mat w(\mat s)$. This time however, the base processes $w_j(\cdot)$ do not have the same distribution. As before, we consider the processes to have a mean of 0 and a variance of 1, but we make them distinct by assigning each $w_j(\cdot)$ a different correlation function $\rho_j(\cdot)$ \citep{schmidt2003bayesian}. We will restrict ourselves to one family of isotropic and stationary correlation functions indexed by distinct range parameters so that $\cov(w_j(\mat s), w_j(\mat t)) = \rho_j(\mat s,\mat t) = \rho(\varphi_j|\mat s - \mat t|)$.

Gaussian distribution properties under linear transformations tell us that the cross-covariance function of the multivariate process $\mat v(\cdot)$ is given by $C(\mat s, \mat t) = \sum_{j=1}^p \rho_j(\mat s,\mat t) \mat a_j\mat a_j^\top$ where $\mat a_j$ is the $j^{\mathrm{th}}$ column of $\mat A$. For a finite sample of size $n$ organized in a column vector $(\mat v(\mat s_1)^\top,\dots,\mat v(\mat s_n)^\top)^\top$ of length $np$, this equates to a multivariate normal distribution of mean 0 and covariance matrix $\sum_{j=1}^p \mat R_j \otimes \mat a_j \mat a_j^\top$ \citep{gelfand2004nonstationary} with $\mat R_j$ being the spatial correlation matrix associated with $(w_j(\mat s_1), \dots,  w_j(\mat s_n))^\top$.

Neither the covariance function nor matrix can be factorized under the non-separable LMC so it is unclear how the structure of the model could be exploited. On the other hand, we know intuitively that generating LMC-distributed random variables does not require the Cholesky factor of the $np \times np$ covariance matrix. We can simulate the base point processes $w_j(\cdot)$ at locations $\mat s_1,\dots,\mat s_n$ from the individual Cholesky factors ($\mat B_j$, say) of the spatial correlation matrices $\mat R_j,j=1\dots,p$ and then apply the linear transformation $\mat A$ at every location. In algebraic terms, we can generate a $p \times n$ LMC distributed matrix $\mat V$ with rows (resp. columns) corresponding to processes (resp. locations) by setting
$$
\mat V = \sum_{j=1}^p \mat a_j \mat z_j \mat B_j^\top,
$$
where $\mat z_j$ are row vectors of independent and standard normal random variables. Just as we can circumvent the factorization of the $np \times np$ covariance matrix when simulating the LMC, we can also avoid computing the inverse and determinant directly when evaluating the likelihood.

In Appendix \ref{proofs}, we detail the following two results which have appeared in \citet{kyzyurova2017uncertainty} under some slightly different notation. The first one tells us how to compute the precision matrix: $(\sum_{j=1}^p \mat R_j \otimes \mat a_j \mat a_j^\top)^{-1} = \sum_{j=1}^p \mat R_j^{-1} \otimes \mat a_j^{-\top} \mat a_j^{-1}$, where $\mat a_j^{-1}$ is the $j^{\mathrm{th}}$ row of the inverse matrix $\mat A^{-1}$. The second one concerns the determinant and says that $\det(\sum_{j=1}^p \mat R_j \otimes \mat a_j \mat a_j^\top) = \det(\mat A)^{2n} \prod_{j=1}^p \det(\mat R_j)$. We go further and express the LMC likelihood in terms of the $p \times n$ matrix $\mat V$. This facilitates the calculation beyond what is shown in \citet{kyzyurova2017uncertainty} by simplifying what would be a $np \times np$ quadratic product under the vectorized form.

\begin{proposition}[The LMC Density]\label{LMCdensRes}
    Let $\mat w_j\indep\cN(\mat 0,\mat R_j), j=1,\dots,p$ be independent $n$-vectors forming the rows of the $p \times n$ matrix $\mat W = (\mat w_1,\dots,\mat w_p)^\top$. The random matrix $\mat V = \mat A \mat W$ (for some $p \times p$ full-rank matrix $\mat A$) has probability density function 
    \begin{align}
    p(\mat V | \mat A,\mat R_{j=1}^p) = \frac{\exp [ -\frac{1}{2} \sum_{j=1}^p \mat a_j^{-1}\mat V \mat R_j^{-1} \mat V^\top \mat a_j^{-\top}]}{(2\pi)^{np/2} |\det(\mat A)|^{n} \prod_{j=1}^p \det(\mat R_j)^{1/2}  }. \label{LMCdens}
    \end{align}

    Equivalently, we write $\mat V \sim \LMC(\mat A,\mat R_{j=1}^p)$ if $\vec(\mat V)$ has the multivariate normal distribution with mean $\mat 0$ and covariance $\sum_{j=1}^p \mat R_j \otimes \mat a_j \mat a_j^\top$ with $\mat A = (\mat a_1,\dots,\mat a_p)$. 
\end{proposition}

Expression \eqref{LMCdens} is analogous (in the case of the general LMC) to the matrix normal density in equation \eqref{matNorm}. We can see that it is again possible to evaluate the likelihood of $\mat V$ without computing the determinant or inverse of any $np \times np$ matrix neither is it necessary to store any such matrix. Computation can be carried out directly on the coregionalization matrix $\mat A$ and the $p$ spatial correlation matrices $\mat R_j$ obtained by applying the functions $\rho(\varphi_j|\cdot|)$ to all pairs of locations in $\mat s_1, \dots, \mat s_n$. In Figure \ref{likeEval}, we empirically compare how likelihood computations scale: we can see the polynomial trend (as a function of $p$) of the naive vectorized form compared to the likelihood formulation in \eqref{LMCdens} which is linear in $p$.

\begin{figure}[h!]
     \centering
     
        \includegraphics[width=\textwidth]{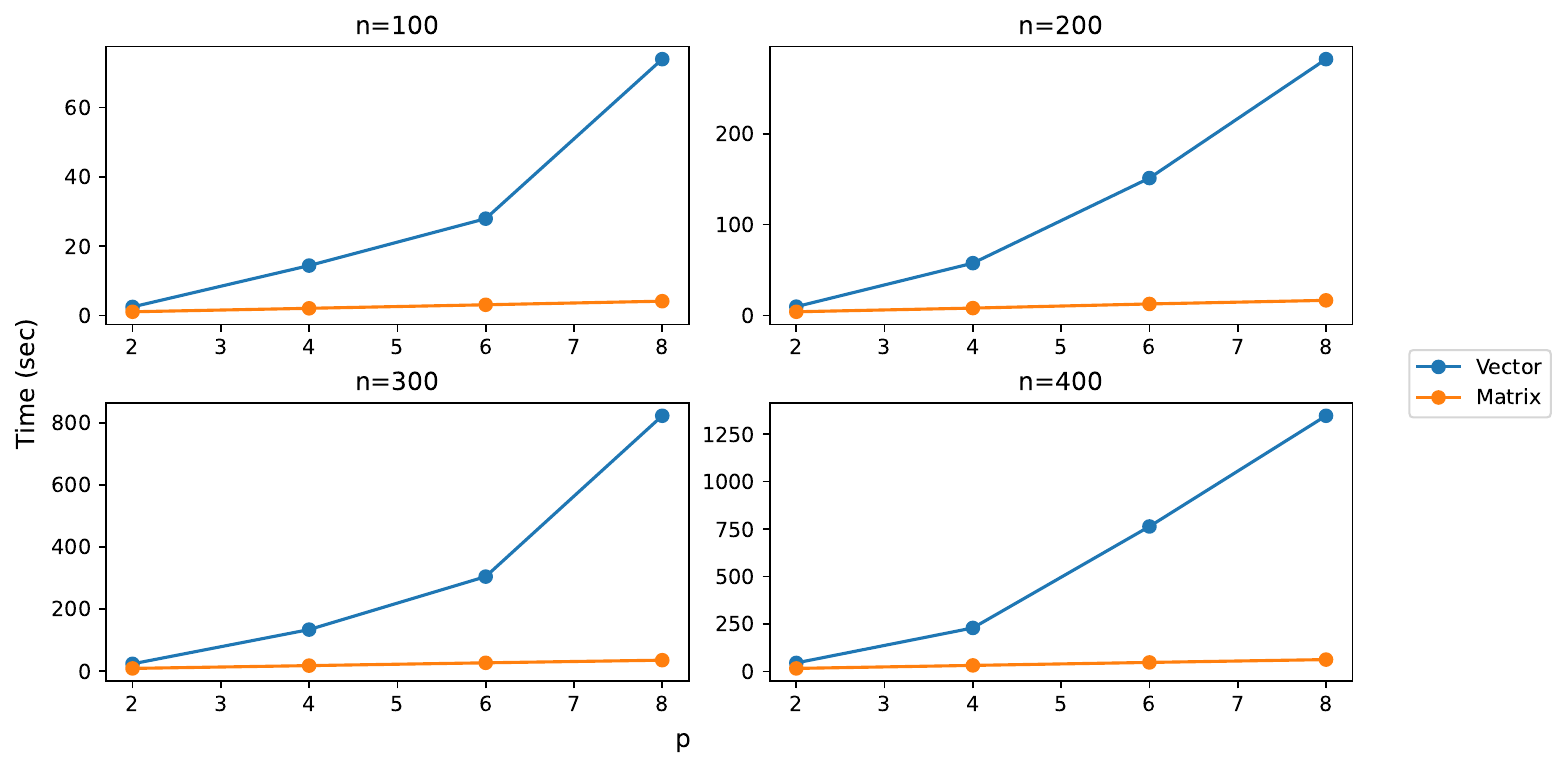}
        \caption{Time elapsed after 500 likelihood evaluations for corresponding $n$ and $p$. The matrix form of the likelihood is the one described in Proposition \ref{LMCdensRes} while the vector form is the naive implementation that computes both the determinant and inverse of $\sum_{j=1}^p \mat R_j \otimes \mat a_j \mat a_j^\top$.}
        \label{likeEval}
\end{figure}

Marginally, the covariance matrix of the process $\mat v(\cdot)$ at any location $\mat s \in \cS$ is $\mat A \mat A^\top$ since the $ w_j(\cdot)$ processes have a variance of 1. However, unlike in the separable case, the finite sample distribution is not the same in general for two different coregionalization matrices $\mat A_1,\mat A_2$ even if $\mat A_1\mat A_1^\top = \mat A_2\mat A_2^\top$ (see Appendix \ref{identif}). There are still however identifiability issues with the matrix $\mat A$ in the covariance matrix $\sum_{j=1}^p \mat R_j \otimes \mat a_j \mat a_j^\top$. The model is equivalent even if we apply a change of sign to any column $\mat a_j$. It is also invariant under any reordering/relabelling of indices $\{1,\dots,p\}$. The general LMC implies a covariance among pairs of scalar observations that is of the form
\begin{align}
   \cov( v_i(\mat s), v_j(\mat t)) = C_{ij}(\mat s, \mat t) = \sum_{k=1}^p a_{ik}a_{jk}\rho(\varphi_k|\mat s - \mat t|), \label{crossCovStruct}
\end{align}
where the quantities $a_{ij}$ are the individual entries in the coregionalization matrix $\mat A$. 



Next, we discuss predicting the multivariate process implied by the LMC at a set $\mat s_\new$ of locations conditional on its values observed at locations $\mat s_\old$. We start by describing the conditional distributions of the LMC. Same as before, we assume $\mat V$ is a $p \times n$ matrix, only this time the matrix is split along the horizontal into two sub-matrices $\mat V_\old$ and $\mat V_\new$ such that $\mat V = (\mat V_\old, \mat V_\new)$. In that case, the spatial covariance matrices $\mat R_j$ are each partitioned between the observed and prediction locations as 
\begin{align}
\mat R_j = \left(\begin{matrix*}[l]
\mat R^{(j)}_\old & \mat R^{(j)}_{\old,\new}\\
\mat R^{(j)}_{\new,\old} & \mat R^{(j)}_\new  
\end{matrix*}\right). \label{fullCorr}
\end{align}
We are interested in the conditional distribution of $\mat V_\new$ given $\mat V_\old$. Importantly, the distribution of $\mat V_\new$ conditional on $\mat V_\old$ preserves the LMC covariance structure. This is detailed in the following result.

\begin{proposition}[Conditional Distributions of the LMC]\label{LMCcondDist}
    The distribution implied by the LMC at locations $\mat s_\new$ conditional on its values $\mat V_\old$ at locations $\mat s_\old$ is also an LMC with coregionalization matrix $\mat A$ and spatial correlation matrices given by $\mat R^{(j)}_\new -\mat R^{(j)}_{\new,\old}{\mat R^{(j)^{-1}}_\old}\mat R^{(j)}_{\old,\new}$. It is centered around the matrix
    \begin{align}
    \sum_{j=1}^p \mat a_j \mat a_j^{-1} \mat V_\old {\mat R^{(j)^{-1}}_\old}\mat R^{(j)}_{\old,\new}.\label{simpleKrig}
    \end{align}
    This is equivalent to applying the inverse transformation $\mat A^{-1} \mat V_\old = \mat W_\old$, computing the conditional $\mat W_\new|\mat W_\old$ across the $p$ independent rows and transforming back as $\mat V_\new = \mat A \mat W_\new$.
\end{proposition}

The proof can be found in Appendix \ref{proofCond}. Proposition \ref{LMCcondDist} has implications for Vecchia type of approximations \citep{vecchia1988estimation,datta2016hierarchical,katzfuss2021general} which are described by conditional independence assumptions. This computational approximation offers likelihood evaluations that are linear in the number $n$ of locations. It also has been shown empirically to be accurate when compared with alternative methods \citep{heaton2019case}. That the LMC covariance structure is preserved under conditioning signifies that the computational benefits we describe here can in principle be used in conjunction with such approximations when both $n$ and $p$ are fairly large. 

\subsection{The triangular specification}\label{triSpec}

A common thread among the statistical literature concerning the LMC  (see \citet{schmidt2003bayesian,foley2008statistical} for example) is to use a lower triangular parametrization for the coregionalization matrix $\mat A$ where diagonal elements are restricted to positive values. It does solve the identifiability issue described in the previous section. Columns cannot change signs anymore because of the positivity of the diagonal. They cannot be reordered either since that would break the triangular structure.

Doing so, however, imposes an asymmetry on the components of the model. The modeling assumptions are different under a reordering of dimension indices $\{1,\dots,p\}$ even though we usually think of such labels as superfluous. It is perhaps better appreciated by looking at the marginal covariance function implied by the lower triangular structure. For each process $ v_j(\cdot),j=1,\dots,p$, we have from \eqref{crossCovStruct} that $C_j(\mat s,\mat t) = \sum_{k=1}^j  a_{jk}^2 \rho(\varphi_k |\mat s - \mat t|)$. Marginal models go up in complexity as indices increase. This limitation is also outlined in \citet{carter2024land}.

As discussed in \citet{schmidt2003bayesian} and \citet{gelfand2004nonstationary}, the lower triangular specification for $\mat A$ admits a computational simplification when describing the processes conditionally as $v_j(\cdot)|v_1(\cdot),\dots,v_{j-1}(\cdot)$. This arises because we can write $\mat A^{-1} \mat v(\mat s) = \mat w(\mat s)$ and $\mat A^{-1}$ is also lower triangular. Hence, for any $j=1,\dots,p$ we can write $v_j(\mat s) = \sum_{\ell=1}^{j-1} \alpha_{j\ell} v_\ell(\mat s) + \sigma_j w(\mat s)$ with $\alpha_{j\ell} = - a^{-1}_{j\ell} / a^{-1}_{jj}$ and $\sigma_j = 1/a^{-1} _{jj}$ ($a^{-1}_{j\ell}$ are elements of the inverse matrix $\mat A^{-1}$). Each value $v_j(\mat s)$ is a linear combination of the previous $v_1(\mat s),\dots,v_{j-1}(\mat s)$ and the $j^{\text{th}}$ underlying process $w_j(\mat s)$. The covariance matrix of process $v_j(\cdot)$ at locations $\mat s_1,\dots,\mat s_n$ conditional on previous $v_\ell(\cdot),\ell=1,\dots,j-1$ simplifies to that of $w_j(\cdot)$ at said locations.

To summarize, the likelihood can be computed in conditional factorization from the inverses and determinants of individual matrices $\mat A, \mat R_1,\dots, \mat R_p$. This is an attractive formulation from a computational standpoint. Generalizing this, we demonstrated in Section \ref{genLMC} that any LMC parametrization benefits from a similar computational simplification. This is fortunate as the triangular approach is limiting and hard to justify. It is a fundamental (rather than circumstantial) property of the LMC to have likelihood evaluations that require a number of operations that is linear in $p$. Computational simplicity is perhaps its most attractive feature.

\subsection{Hierarchical Models}\label{hierModels}

In this section, we employ the LMC as the latent component in a hierarchical Gaussian model. We first describe what is a standard problem when discussing computational scalability in any Gaussian process model. Suppose you observe a process $\mat Y(\mat s) =  \mat v(\mat s) + \mat \epsilon(\mat s)$ where $\mat v(\cdot)$ is some zero mean process accounting for spatial dependence in $\cS\subset \mathds R^d$ space while $\mat \epsilon(\cdot)$ is some form of white noise. The multivariate $\cN(0,\mat \Sigma)$ distribution implied by process $\mat v(\cdot)$ observed at locations $\mat s_1, \dots, \mat s_n$ might benefit from some useful computational tricks when computing the likelihood. This could be the result of employing an approximation such as the nearest neighbor approach of \citet{datta2016hierarchical} that provides a factorization (detailed in \citet{finley2019efficient}) of the precision matrix $\mat \Sigma^{-1}$. It could also be because the structure of the covariance matrix $\mat\Sigma$ can be exploited, as is the case for the LMC. In either case, we face the same problem: the finite-dimensional distribution of $\mat Y(\cdot)$ observed at $n$ locations will be of the form $\cN(0,\mat  \Sigma +  \mat D)$ for some diagonal matrix $\mat D$. Since $\mat D$ is of full rank, there is no simple way to evaluate the Gaussian likelihood without inverting or computing the determinant of any $np \times np$ matrix (in the multivariate setting). In other words, working with the distribution of the collapsed Gaussian process $\mat Y(\cdot)$ forces you to relinquish any computational schemes that applied to $\mat v(\cdot)$. 

In a Bayesian context, there is always the option of sampling the latent random field $\mat v(\cdot)$ and model parameters in turn in a Gibbs sampling approach. However, this often leads to Markov chains that exhibit high autocorrelation \citep{papaspiliopoulos2007general,murray2010slice,yu2011center}. We explore strategies to alleviate this problem for the specific case of the LMC in Appendix \ref{interSec}. The distribution of $\mat v$ conditional on $\mat y$ is multivariate normal with covariance matrix $(\mat \Sigma^
{-1} + \mat D^{-1})^{-1}$ which is also cumbersome to work with. Nevertheless, sampling from the full conditionals of the $np$-dimensional distribution of $\mat v| \mat y$ can be accomplished without the need to invert or factorize any $np \times np$ matrix. Details are given in Appendix \ref{datAugStep}.

In the context of the LMC, the finite sample equation has the form $\mat Y = \mat V + \mat \epsilon$ where $\mat V \sim \LMC(\mat A,\mat R_{j=1}^p)$ and $\mat \epsilon \sim\cM \cN (\mat 0, \mat D,\mat I)$. The $p\times p$ matrix $\mat D=\text{diag}\{\tau_j,j=1,\dots,p\}$ contains the distinct observational variances associated with each of the $p$ processes. We omit the usual mean structure to focus our discussion on the sampling of covariance parameters. The LMC structure of $\mat V$ cannot be exploited in order to derive the covariance of $\mat Y$ even though $\mat \epsilon$ is itself an LMC distributed matrix normal. In general, the sum of LMC-distributed matrices is not itself distributed according to the LMC. The same is true for the standard matrix normal distribution described by Definition \ref{matNorm}.

The $2np$-dimensional normal density of the complete data $\mat Y,\mat V$ is of the form:
\begin{align}
p(\mat Y, \mat V | \mat A, \mat R_{j=1}^p, \mat D) = \frac{\exp[-\frac{1}{2} \sum_{j=1}^p \mat a_j^{-1}\mat V \mat R_j^{-1} \mat V^\top \mat a_j^{-\top} -\frac{1}{2} \trace\{(\mat V - \mat Y)^\top \mat D^{-1} (\mat V - \mat Y)\}]}{ (2\pi)^{np} |\det(\mat A)|^{n} \prod_{j=1}^p \det(\mat R_j)^{1/2}\prod_{j=1}^p \tau_j^{n/2}}. \label{compCent}
\end{align}
Because we used the alternative formulation described in Section \ref{genLMC} for the latent LMC $\mat V$, expression \eqref{compCent} above can be evaluated without the need to carry out computations on any $np\times np$ matrices.

The complete data likelihood in \eqref{compCent} describes a hierarchical model for $\mat Y$ that is centered at $\mat V$ with independent error terms $\mat \epsilon$. Alternatively, we could also write $\mat Y = \mat A \mat W + \mat \epsilon$ where $\mat W \sim \LMC(\mat I,\mat R_{j=1}^p)$. The rows $\mat w_1, \mat w_2, \dots, \mat w_p$ of $\mat W$ are the $p$ independent spatial processes underlying the LMC. In this context, the latent process $\mat W$ is {\it a priori} independent of the coregionalization matrix $\mat A$. In essence, since the latent process $\mat V$ is not observed, we can equivalently think of the complete data as $(\mat Y, \mat W)$. This leads to a joint likelihood of the form 
\begin{align}
p(\mat Y, \mat W | \mat A, \mat R_{j=1}^p, \mat D) = \frac{\exp[-\frac{1}{2} \sum_{j=1}^p \mat w_j^\top \mat R_j^{-1}\mat w_j -\frac{1}{2} \trace\{(\mat A \mat W - \mat Y)^\top \mat D^{-1} (\mat A \mat W - \mat Y)\}]}{ (2\pi)^{np} \prod_{j=1}^p \det(\mat R_j)^{1/2} \prod_{j=1}^p \tau_j^{n/2}}.\label{compWhite}
\end{align}
Expression \eqref{compWhite}, same as \eqref{compCent}, factorizes in individual terms containing inverses or determinants of potentially large correlation matrices $\mat R_j$. Those represent the crux of the computation when evaluating the likelihood. It seems we can benefit from similar computational advantages to those described in Section \ref{genLMC} by simply changing the parametrization. 

However, in the context where the independent variance terms $\tau_j$ along the diagonal of $\mat D$ are small, updates to $\mat A$ conditional on fixed $\mat W$ will tend to be very concentrated around the current value. Gibbs samplers that alternate updates to the latent variables $\mat W$ and model parameters $\mat A, \varphi_{j=1}^p, \tau_{j=1}^p$ are likely to exhibit high autocorrelation in the components of $\mat A$ for reasons that are well outlined in \citet{murray2010slice}. In other words, such a formulation of the problem is unlikely to perform very well when the signal-to-noise (StN) ratio is relatively high. In return, it is likely to perform better than the centered parametrization of \eqref{compCent} in the context where the StN ratio is low. However, such a scenario is not auspicious for statistical analysis in general. In Appendix \ref{interSec}, we explore interweaving strategies that can, in principle, benefit from the advantages of both parametrizations.

\section{Sparse LMC}\label{sparseApp}

As a function of dimensionality, the LMC introduces $p^2$ parameters. It makes for a model that is not parsimonious and this might cause overfitting (in the sense of deteriorating predictions at out-of-sample locations) in cases where $p$ is moderately large. There could also exist independencies among the processes; it would perhaps be expected in high dimensional settings. We address those issues by designing a prior that puts probability mass on coregionalization matrices that have structural zeros. We introduce the concept of a mask $\mat M$ which is a binary matrix of the same size as $\mat A$. The true coregionalization matrix is then the element-wise (Hadamard) product $\mat A \circ \mat M$. This matrix still needs to be of full rank such that none of the $p$ processes end up being deterministic linear combinations of the others. The triangular specification of the LMC discussed in Section \ref{triSpec} provides an example of a coregionalization matrix with structural zeros (everywhere over the main diagonal in this case). We argued that a lower triangular $\mat A$ is a model hypothesis that is hard to justify. In comparison, our sparse approach reduces the dimensionality of the parameter space by learning the structure of the coregionalization matrix in a principled and data-driven approach.

Same as before, we put i.i.d. normal priors on the elements of $\mat A$ which guarantees that $\det (\mat A) \neq 0$ with probability 1. However, this might not hold when setting some elements to zero. For example, setting an entire row or column of $\mat M$ to zero will always lead to $\mat A \circ \mat M$ being singular. For any given binary matrix $\mat M$ and assuming we assign a distribution to the elements of $\mat A$ that is absolutely continuous with respect to the $p^2$-dimensional Lebesgue measure, we have that $P(\det(\mat A \circ \mat M) = 0)$ is either 0 or 1. We detail below the process by which one can verify if a model, defined by the binary masking matrix $\mat M$, is admissible in the sense that the matrix $\mat A \circ \mat M$ will be of full rank with probability 1.

We start with the Laplace expansion along the $i^{\text{th}}$ row of the determinant
\begin{align}
\det( \mat A \circ \mat M) = \sum_{j=1}^p  a_{ij}  m_{ij} C_{ij} \label{laplaceExp}
\end{align}
where $C_{ij}$ is the $(i,j)$ cofactor: the determinant of the matrix obtained by removing the $i^{\text{th}}$ row and $j^{\text{th}}$ column of $\mat A \circ \mat M$ multiplied by $(-1)^{i+j}$. Equation \eqref{laplaceExp} will be equal to 0 if and only if every summand is identically 0.

The determinant can be expanded from any row or column. If any of those consists entirely of 0s, then expression \eqref{laplaceExp} will be equal to 0 with probability 1. This will be a stopping condition for the recursive algorithm designed to evaluate the admissibility of a model $\mat M$. The other stopping state is when there are globally less than $p$ zeros. In that case, expression \eqref{laplaceExp} is almost surely not equal to 0. Indeed, this is trivially true for $p=1$ and can be shown to be true for any $p$ by induction using the Laplace expansion. Equation \eqref{laplaceExp} allows us to transform a statement about a $p \times p$ matrix into verifying the same condition on smaller $(p-1) \times (p-1)$ matrices. The determinant will be identically zero if and only if all the determinants of the smaller matrices are identically 0 (those corresponding to the non-zero entries of the row or column we chose to expand from). A good strategy is thus to always expand the determinant by the row or column that contains the most 0s to minimize the number of recursive calls. 

This allows us to formulate a hierarchical prior where the model $\mat M$ is selected at the first stage, regardless of $\mat A$, among admissible ones. The respective entries of the coregionalization matrix are assigned at the later stage. The prior on the masking matrix is assumed to take the form
\begin{align}
p(\mat M) \propto \pi^{\sum m_{ij}} (1-\pi)^{p^2 - \sum m_{ij}}, \label{priorModel}
\end{align}
where $\sum m_{ij}$ is the number of 1s contained in $\mat M$ and $\pi \in [0,1]$ is a hyperparameter controlling prior belief on the degree of sparsity. It is essentially the prior probability for an entry of $\mat M$ to be non-zero, subject to the restrictions described above. We use the framework of reversible jump MCMC \citep{green1995reversible} to conceive proposals to the coregionalization matrix $\mat A \circ \mat M$ that allow trans-dimensional updates. At every step, we either propose to replace a non-zero entry with a 0 or to assign a new value to a previously null entry. We chain $p$ of those proposals at each iteration of the Markov chain. The computational overhead is negligible in the context where we also estimate range parameters $\varphi_{j=1}^p$, the updates to which are orders of magnitude more costly (see Appendix \ref{algosDetails} for more details). The model parameters $\mat A,\mat M, \varphi_{j=1}^p, \tau_{j=1}^p$ and the unobserved random fields $\mat V$ are updated successively in a Gibbs sampler. Next, we conduct an experiment designed to evaluate the effect on predictive accuracy of our sparsity-inducing approach. 

\subsection{Simulation Study} \label{simStud}

Model parameters for the sparse LMC include $\mat A, \mat M, \varphi_{j=1}^p, \tau_{j=1}^p$. The prior distribution on the submodel defined by the $p \times p$ binary matrix $\mat M$ is described by equation \eqref{priorModel}. We use $\pi=1/p$ throughout as our prior belief on the degree of sparsity. We put independent Gaussian priors on the elements of the coregionalization matrix $\mat A$. We use the exponential correlation function $\rho(\varphi_j|\mat s-\mat t|) = \exp(-\varphi_j|\mat s-\mat t|)$ for the underlying process $w_j(\cdot)$ and the range parameters $\varphi_1,\dots,\varphi_p$ are assigned uniform priors on $[3,30]$. Roughly speaking, this corresponds to a practical range at distances between 0.1 and 1 (points are observed in the unit square). We use this informative prior on the range parameters because they cannot be consistently estimated along with the marginal variances contained in the coregionalization matrix \citep{zhang2004inconsistent}. Finally, we assign independent inverse Gamma priors with high variance to the parameters $\tau_1,\dots,\tau_p$ to obtain conjugate full conditional updates. We employ the complete data likelihood of $\mat Y, \mat V$ described by equation \eqref{compCent}. Additional details on the Bayesian setup for this study are provided in Appendix \ref{MCMCsetup}.

We start by simulating a set of $n=100$ irregularly spaced locations from the uniform distribution on the unit square. We keep those locations fixed throughout the study. For each of $p=2,3,4,5$, we generate 100 realizations of the multivariate GP. The model is simulated according to the LMC for each of two examples. In the first, the coregionalization matrix $\mat A \circ \mat M$ is full: it has no structural zeros. All processes are interdependent. In the second example, every entry of $\mat M$ is zero except the main diagonal hence the $p$ resulting processes are mutually independent. In both cases, the rows of the matrix $\mat A$ are scaled so that each process $v_j(\cdot)$ has unit marginal variance. For the full $\mat M$ example, the coregionalization matrix is filled with the value $1/\sqrt{p}$ with entries above the diagonal being negative. For the diagonal example, $\mat A \circ \mat M$ is simply the identity matrix. We add an observational noise of variance $\tau_j = 0.25, j=1,\dots,p$ to the $p$ components at each location. For the range parameters $\varphi_{j=1}^p$, we split the $[5,25]$ interval with equal gaps on the logarithmic scale so for example we use (5, 7.48, 11.18, 16.72, 25) when $p=5$. This is to ensure the $p$ underlying processes are as distinguishable as possible while having range inside the [3,30] support. 

We run two MCMC inference algorithms on each LMC repetition. The first fits a regular LMC (without model selection) and the second fits the sparse LMC using reversible jumps. Both algorithms are described in Appendix \ref{algosDetails}. We run each algorithm for 2000 iterations, keeping only the last 1000. For each posterior sample of the latent random fields $\mat V$ and LMC parameters $\mat A, \mat M, \varphi_{j=1}^p$, we interpolate the processes to a regular $10\times10$ grid covering the domain. We compute the root mean square error (averaged over the grid locations and also across posterior samples) between the predictions and the held-out true values. We look at the difference in prediction error between the sparse and regular LMC fit. The distribution (across LMC repetitions) of this difference for the full and diagonal examples are respectively showcased in Figures \ref{full_exp} and \ref{diag_exp}.

\begin{figure}[h!]
     \centering
     \begin{subfigure}[b]{0.49\textwidth}
         \centering
         \includegraphics[width=\textwidth]{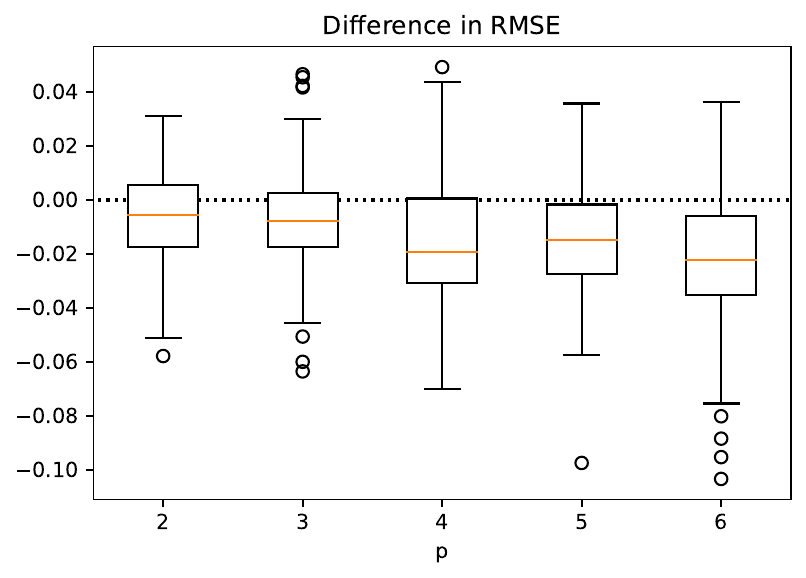}
     \end{subfigure}
     \hfill
     \begin{subfigure}[b]{0.46\textwidth}
        \centering
         \includegraphics[width=\textwidth]{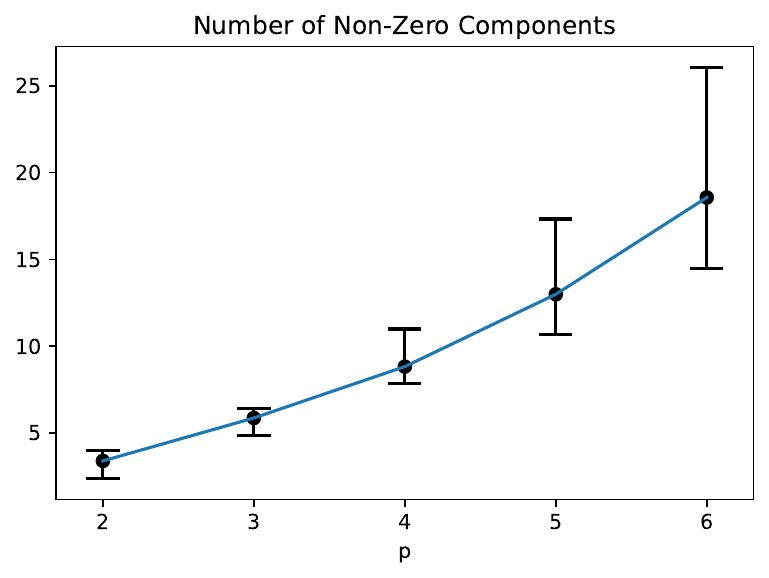}
    \end{subfigure}
        \caption{Comparison of our sparse method with the standard approach for the case of a full (no zeros) coregionalization matrix. On the left, the difference (sparse minus standard) in RMSE computed on each realization. On the right, we illustrate the posterior-mean number of non-zero components in the matrix $\mat M$. The median (point) and [0.05,0.95] quantiles are computed across the 100 simulated LMC realizations.}
        \label{full_exp}
\end{figure}

\begin{figure}[h!]
     \centering
     \begin{subfigure}[b]{0.49\textwidth}
         \centering
         \includegraphics[width=\textwidth]{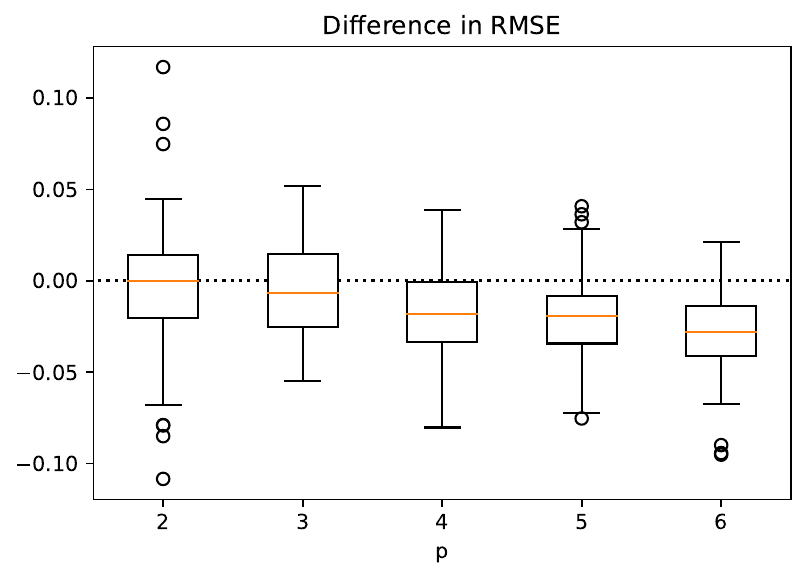}
     \end{subfigure}
     \hfill
     \begin{subfigure}[b]{0.46\textwidth}
        \centering
         \includegraphics[width=\textwidth]{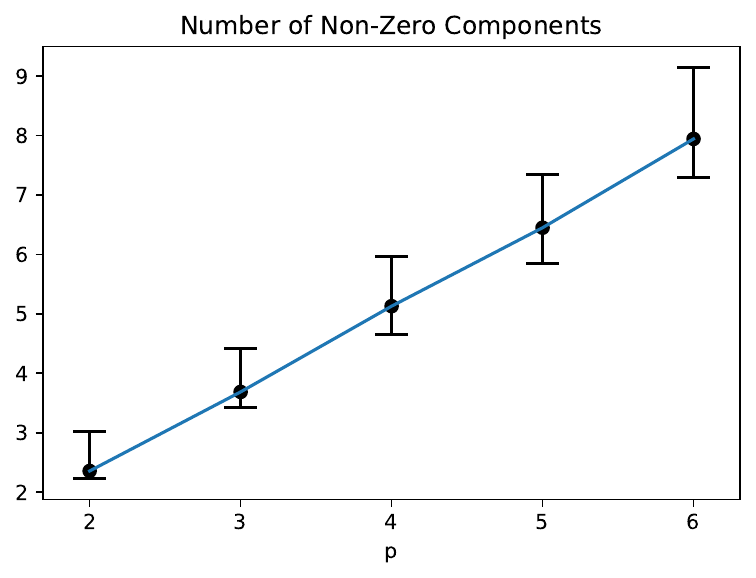}
    \end{subfigure}
        \caption{Comparison of our sparse method with the standard approach for the case of a diagonal coregionalization matrix (independent components). On the left, the difference (sparse minus standard) in RMSE computed on each realization. On the right, we illustrate the posterior-mean number of non-zero components in the matrix $\mat M$. The median (point) and [0.05,0.95] quantiles are computed across the 100 simulated LMC realizations.}
        \label{diag_exp}
\end{figure}

In lower dimension $p$, there is little difference in using either of the two estimation methods. However, as dimensions increase, predictions from the sparse model appear more accurate than those under the standard approach with no structural zeros. This is certainly expected for the diagonal example as there are only $p$ non-zero entries in the coregionalization matrix. It is perhaps more remarkable for the full $\mat M$ case where a misspecified but sparser model better generalizes to out-of-sample predictions. This is because the number $p^2$ of coregionalization parameters to estimate rapidly becomes unwieldy when compared to the number of observed locations. Our sparse approach facilitates the use of the LMC in such scenarios. We can also see on the right side of Figures \ref{full_exp} and \ref{diag_exp} that our method manages to curb the growth of parameters especially in the diagonal case where we observe a seemingly linear relation. 

\begin{table}[h!]
\centering\footnotesize
\begin{tabular}{|c|ccccc|}
\hline
&&&$p$&&\\
Method&2&3&4&5&6\\
\hline
Sparse &15.687 & 18.658 & 21.607 & 24.706 & 27.869 \\
Standard&15.741 & 18.477 & 21.205 & 24.038  & 26.906\\
\hline
\end{tabular}
\vspace{0.2cm}
\caption{Mean time (in seconds) elapsed per model fitting for the different scenarios.}\label{compTime}
\end{table}
In Table \ref{compTime}, we compare the computing time of the sparse and standard specifications. Our approach adds very little computing time when compared to the traditional LMC. This is because recomputing the likelihood upon a modification of the coregionalization matrix is negligible under our formulation \eqref{LMCdens} of the LMC density (see appendix \ref{algosDetails} for more details). In the next section, we use our sparse methodology on a real dataset to showcase the inferential properties of our model.

\subsection{Analysis of California Air Data} \label{anCaliData}

In this section, we analyze a data set consisting of measurements in the concentration of four pollutants: carbon monoxide (CO), nitric oxide (NO), nitrogen dioxide (NO$_2$) and ozone (O$_3$). There are no covariates included. Our interest is mainly in inferring the cross-covariance structure. We use the reversible jump sampler for the coregionalization matrix $\mat A \circ \mat M$. This will help us identify potential independencies among the measured quantities. We analyze the data on the log scale to achieve approximate normality. The measurements are made at monitoring sites spread across the state of California. Those locations, projected to the unit square, are illustrated in Figure \ref{monSites}. The data is available publicly on the California Air Resources Board website. We chose the date of 26/04/2002 as it had the highest number of sites where all 4 components were measured with 83. We also included 48 stations that had partial observations (at least one measurement missing) for a total number of 131. The proportion of missing values for the components CO, NO, NO$_2$ and O$_3$ were respectively 0.25, 0.10, 0.09 and 0.05. Data from this database was also analyzed in \citet{schmidt2003bayesian} although in the summer rather than the spring and without the ozone component. Nevertheless, the resulting inference we present here, while not identical, seems to be coherent with their findings in terms of the magnitude of correlations among the three components CO, NO and NO$_2$.

\begin{figure}[h!]
     \centering
     
        \includegraphics[width=0.4\textwidth]{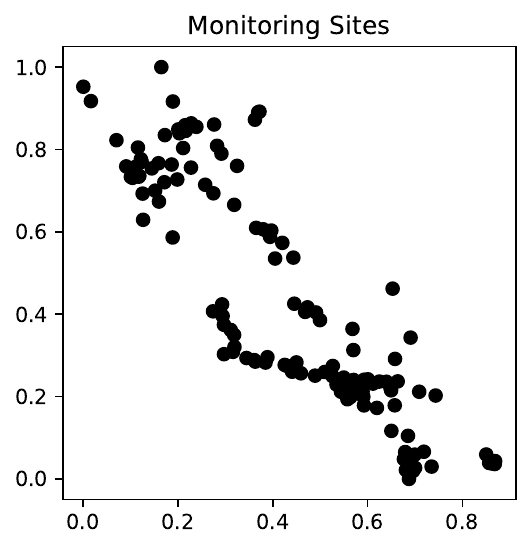}
        \caption{Projected locations of the 131 monitoring stations used in the analysis of the California air data.}
        \label{monSites}
\end{figure}

We assign independent Gaussian priors to the non-zero components of $\mat A$. The range components $\varphi_{j=1}^p$ are constrained to the uniform [3,30] distribution, {\it a priori}. The observational noise variance parameters $\tau_{j=1}^p$ are given inverse gamma prior distributions. The log concentrations of each of the four pollutants do not appear to share a common mean hence we add a component-specific mean structure to the LMC. The $p$-dimensional mean vector $\mat\mu$ is assigned a Gaussian distribution with diagonal covariance to obtain conjugate normal full conditional updates. Finally, we use $1/p=0.25$ for the hyperparameter $\pi$ that controls the degree of sparsity. Additional details on the MCMC setup can be found in Appendix \ref{MCMCsetup}. We run the algorithm for 200,000 iterations, discarding the first 50,000 as burn-in.

The posterior samples of the covariance and correlation matrices are summarized in Figure \ref{covCorrmats}. We can see some of the credible intervals have boundaries that are exactly 0. This is a consequence of assigning prior probability mass at 0 for the elements of the coregionalization matrix $\mat A \circ \mat M$. Another interesting component of our sparsity-inducing approach to the LMC is that we can evaluate the posterior probability of two processes to be independent. Those probabilities are illustrated in Figure \ref{probInd} along with the distribution of non-zero elements in the matrix $\mat M$. It seems to potentially indicate that the CO pollutant could be uncorrelated with the other three. Further evidence would be needed before reaching this conclusion, but this is not the object of this paper. Modeling pairwise dependencies among the other components seems to be worth it since they have probability 0 of being independent. In Figure \ref{probIndDum}, we show the posterior probabilities of independence in an alternate scenario where we simulate an independent fifth component (D). We can see that our model appropriately captures this independency.

\begin{figure}[h!]
     \centering
     \footnotesize 
\begin{tabular}{|c|cccc|}
\hline
     & CO & NO & NO$_2$ & O$_3$\\
     \hline
CO &\begin{tabular}[c]{@{}c@{}}0.13\\ (0.00;0.41)\end{tabular}   &     &     &   \\
NO &\begin{tabular}[c]{@{}c@{}}0.17\\ (0.00;0.41)\end{tabular}   & \begin{tabular}[c]{@{}c@{}}0.65\\ (0.36;1.05)\end{tabular}    &     & \\
NO$_2$ &\begin{tabular}[c]{@{}c@{}}0.09\\ (0.00;0.30)\end{tabular}   & \begin{tabular}[c]{@{}c@{}}0.46\\ (0.28;0.75)\end{tabular}    & \begin{tabular}[c]{@{}c@{}}0.52\\ (0.32;0.92)\end{tabular}    &  \\
O$_3$ & \begin{tabular}[c]{@{}c@{}}-0.02\\ (-0.07;0.00)\end{tabular} & \begin{tabular}[c]{@{}c@{}}-0.10\\ (-0.17;-0.06)\end{tabular} & \begin{tabular}[c]{@{}c@{}}-0.09\\ (-0.16;-0.05)\end{tabular} & \begin{tabular}[c]{@{}c@{}}0.03\\ (0.01;0.05)\end{tabular}  \\
\hline
\end{tabular}

    \caption{ Posterior median of the marginal (at each location) covariance matrix along with the 0.05 and 0.95 quantiles in parenthesis.}
    \label{covCorrmats}
\end{figure}

\begin{figure}[h!]
     \centering
     \begin{subfigure}[b]{0.49\textwidth}
     \footnotesize
     \centering
     \begin{tabular}{|c|cccc|}
     \hline
    & CO & NO & NO$_2$ & O$_3$\\
     \hline&&&&\\
CO&0.0000 &  &  &  \\&&&&\\
NO&0.2243 & 0.0000 &  &  \\&&&&\\
NO$_2$&0.3785 & 0.0000 & 0.0000 &  \\&&&&\\
O$_3$&0.3804 & 0.0000 & 0.0000 & 0.0000 \\&&&&\\
\hline
\end{tabular}
\vspace{0.5cm}
\caption{Pairwise Probability of Independence}
     \end{subfigure}
     \hfill
     \begin{subfigure}[b]{0.49\textwidth}
     \centering
     \includegraphics[height=0.52\textwidth]{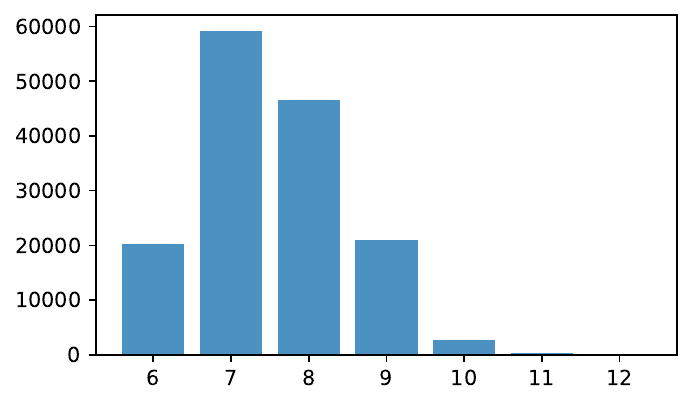}
\caption{Number of ones in $\mat M$ (out of $p^2 = 16$ entries)}
    \end{subfigure}
    \caption{}
    \label{probInd}
\end{figure}

\begin{figure}[h!]
     \centering\footnotesize 
        \begin{tabular}{|c|ccccc|}
     \hline
    & CO & NO & NO$_2$ & O$_3$&D\\
     \hline&&&&&\\
CO&0.0000 &  &  &  & \\&&&&&\\
NO&0.1187 & 0.0000 &  &  &  \\&&&&&\\
NO$_2$&0.2673 & 0.0000 & 0.0000 &  &  \\&&&&&\\
O$_3$&0.3776 & 0.0000 & 0.0000 & 0.0000  & 
\\&&&&&\\
D&0.6037 & 0.6066 & 0.8207 & 0.8623 & 0.0000\\&&&&&\\
\hline
\end{tabular}
    
        \caption{Pairwise probability of independence with the component D being an independent Gaussian process with mean -3.92, unit variance and range parameter equal to 10.}
        \label{probIndDum}
\end{figure}

\section{Discussion}

We revisited the LMC as a constructive approach to specifying cross-covariance functions, deriving some important characteristics in the process. Starting with the separable specification of the LMC and its link with the matrix normal distribution, we were able to make an analogous connection between the more general LMC and a new kind of Gaussian distribution for matrices. Doing so exposed a covariance structure that can be exploited to obtain likelihood evaluations in a number of operations that is linear in the dimensionality $p$. Computational scalability might just be the strongest suit of the LMC, a fact that has not been taken advantage of in the literature about multivariate Gaussian processes. Since this is not an approximation, there are only benefits to using the formulation we describe in Section \ref{genLMC}. 

Beyond the plain computational speedups, there are other advantages to the matrix normal representation of the LMC. Foremost, the coregionalization matrix $\mat A$ is not entangled with the spatial correlation matrices $\mat R_1,\dots,\mat R_p$ in the likelihood. In particular, an update to $\mat A$ does not require computing the inverse and determinant of any $np \times np$ sized matrix. Even though it is generally small in comparison with the spatial correlation matrices, the matrix $\mat A$ contains more parameters and can become difficult to estimate whenever $p$ is moderately large. However, updates to $\mat A$ come cheap so we can incorporate more involved samplers for it without much of an effect on the total computing time. We exploit this fact in a new model specification for the LMC where probability mass is assigned to coregionalization matrices with structural zeros. Doing so allows us to curb the quadratic growth (as a function of $p$) of the number of parameters associated with the LMC. This type of regularization benefits the accuracy of predictions as was observed on out-of-sample data in Section \ref{sparseApp}. Moreover, our sparse approach to the LMC can provide interesting insights by discovering potential independencies among the $p$ modeled components.

In terms of future developments and applications, we can think of a few. First, since the computational structure is preserved under conditioning, we would like to see it exploited and fully integrated into the framework of Vecchia approximations like the nearest neighbor approach. This would in principle be a straightforward application of Propositions \ref{LMCdensRes} and \ref{LMCcondDist}. It would form a class of highly scalable (as a function of both $n$ and $p$) multivariate Gaussian models. Also, we used the flexibility of the matrix normal formulation to propose a sparse version of the LMC along with the appropriate trans-dimensional MCMC algorithm. We believe researchers can get even more creative and design samplers that facilitate estimation in particular contexts and applications. Finally, while we focus on spatial covariance functions throughout this paper, any structure on the underlying independent processes can be employed. For example, the dependence in time series applications can often be modeled using Markovian assumptions which induce a likelihood that can be decomposed into a series of easier evaluated factors. Introducing cross-process dependence through a linear transformation (applied at each time point) of independent processes ({\it à la} LMC) would, in principle, imply a matrix normal distribution with likelihood evaluations that scale linearly in both dimensions.


\subsubsection*{Funding}
Alie, Stephens and Schmidt acknowledge support from the Natural Sciences and Engineering Research Council of Canada (NSERC). Alie also acknowledges the support from the Fonds de recherche du Québec – Nature et technologies (FRQNT).

\bibliographystyle{plainnat}  
\bibliography{references}  

\appendix
 
 \section{Proof of Proposition \ref{LMCdensRes}} \label{proofs}

We can think of a vectorized form $\mat w' = (\mat w_1,\mat w_2,\dots,\mat w_p)^\top$ for the base processes where we stack the $p$ $n$-dimensional processes instead of ordering by locations. This is how it is presented in \citet{jin2007order} in the context of multivariate conditional auto-regressive models rather than continuously indexed spatial data. The authors get fairly far along the argument we use here.

Since we start with independent processes, the $np \times np$ covariance matrix of $\mat w'$ is block diagonal with blocks consisting of the $n\times n$ correlation matrices $\mat R_1,\dots,\mat R_p$. In the usual representation (ordered by locations), we would consider the $np$-dimensional vector $\mat w = (\mat w(\mat s_1)^\top, \mat w(\mat s_2)^\top,\dots,\mat w(\mat s_n)^\top)^\top$. Of course, $\mat w$ is obtained by a permutation of the elements of $\mat w'$ so there exists a permutation matrix $\mat P$ for which $\mat w = \mat P \mat w'$ and $\var(\mat w) = \mat P \diag\{\mat R_1,\dots,\mat R_p\} \mat P^\top$.

The multivariate process $\mat v = (\mat v(\mat s_1)^\top, \mat v(\mat s_2)^\top,\dots,\mat v(\mat s_n)^\top)^\top$ with cross-dependence is obtained from $\mat w$ by transforming at each location $\mat s$ according to $\mat v(\mat s) = \mat A \mat w (\mat s)$. For the vectorized forms, this means that $\mat v$ is obtained by multiplying a block diagonal matrix consisting of $n$ $\mat A$ matrices along the diagonal with $\mat w$. Putting all the above together, we can write the covariance of the vectorized $\mat v$ as
\begin{align}
\var(\mat v) = 
\left(\begin{matrix}
\mat A & \mat 0 & \dots & \mat 0\\
\mat 0 & \mat A & \dots & \mat 0\\
\vdots & \vdots & \ddots & \vdots\\
\mat 0 & \mat 0 & \dots & \mat A
\end{matrix}\right) 
\mat P
\left(\begin{matrix}
\mat R_1 & \mat 0 & \dots & \mat 0\\
\mat 0 & \mat R_2 & \dots & \mat 0\\
\vdots & \vdots & \ddots & \vdots\\
\mat 0 & \mat 0 & \dots & \mat R_p
\end{matrix}\right) 
\mat P^\top
\left(\begin{matrix}
\mat A^\top & \mat 0 & \dots & \mat 0\\
\mat 0 & \mat A^\top & \dots & \mat 0\\
\vdots & \vdots & \ddots & \vdots\\
\mat 0 & \mat 0 & \dots & \mat A^\top
\end{matrix}\right). \label{blockCov}
\end{align}

\begin{lemma} \label{invLemma}
    The inverse for the covariance matrix of the vectorized form $(\mat v(\mat s_1)^\top,\dots,\mat v(\mat s_n)^\top)^\top$ of the LMC can be computed as 
    $$
    \left(\sum_{j=1}^p \mat R_j \otimes \mat a_j \mat a_j^\top\right)^{-1} = \sum_{j=1}^p \mat R_j^{-1} \otimes \mat a_j^{-\top} \mat a_j^{-1},
    $$
    where $\mat a_j^{-1}$ is the $j^{\text{th}}$ row of the inverse matrix $\mat A^{-1}$.
\end{lemma}
\begin{proof}
This can be shown by direct verification. Alternatively, from equation \eqref{blockCov} we can guess the inverse 
$$
\left(\begin{matrix}
\mat A^{-\top} & \mat 0 & \dots & \mat 0\\
\mat 0 & \mat A^{-\top} & \dots & \mat 0\\
\vdots & \vdots & \ddots & \vdots\\
\mat 0 & \mat 0 & \dots & \mat A^{-\top}
\end{matrix}\right) 
\mat P
\left(\begin{matrix}
\mat R_1^{-1} & \mat 0 & \dots & \mat 0\\
\mat 0 & \mat R_2^{-1} & \dots & \mat 0\\
\vdots & \vdots & \ddots & \vdots\\
\mat 0 & \mat 0 & \dots & \mat R_p^{-1}
\end{matrix}\right) 
\mat P^\top
\left(\begin{matrix}
\mat A^{-1} & \mat 0 & \dots & \mat 0\\
\mat 0 & \mat A^{-1} & \dots & \mat 0\\
\vdots & \vdots & \ddots & \vdots\\
\mat 0 & \mat 0 & \dots & \mat A^{-1}
\end{matrix}\right)
$$
which is exactly the claimed solution. 
\end{proof}

\begin{lemma}
    The determinant for the covariance matrix of the vectorized form $(\mat v(\mat s_1)^\top,\dots,\mat v(\mat s_n)^\top)^\top$ of the LMC can be computed as  
    $$
    \det\left(\sum_{j=1}^p \mat R_j \otimes \mat a_j \mat a_j^\top\right) = \det(\mat A)^{2n} \prod_{j=1}^p \det(\mat R_j).
    $$
 \end{lemma}
\begin{proof}
Evident from equation \eqref{blockCov} since the determinant of the permutation matrix $\mat P$ is either 1 or -1.
\end{proof}

The final step to prove Proposition \ref{LMCdensRes} is to express the $p$ quadratic products $\mat v^\top \left( \mat R_j^{-1} \otimes \mat a_j^{-\top} \mat a_j^{-1} \right) \mat v$ using the matrix form $\mat V$ of the LMC. The $np$-vector $\mat v$ is obtained by stacking the columns of $\mat V$ hence we write $\mat v = \vec(\mat V)$. We can then use the well known identities $(\mat C^\top \otimes \mat A) \vec(\mat B)  = \vec(\mat A \mat B \mat C)$ and $\vec(\mat A)^\top\vec(\mat B) = \trace(\mat A^\top \mat B)$ in this order to obtain
$$
\vec(\mat V)^\top \left( \mat R_j^{-1} \otimes \mat a_j^{-\top} \mat a_j^{-1} \right) \vec(\mat V) = \trace(\mat V^\top \mat a_j^{-\top} \mat a_j^{-1} \mat V \mat R_j^{-1})
$$
and use the cyclic property of the trace to express it as the scalar $\mat a_j^{-1} \mat V \mat R_j^{-1}\mat V^\top \mat a_j^{-\top}$.

The argument can be recycled for models where the coregionalization matrix varies over locations. Such is the case for the non-stationary LMC proposed by \citet{gelfand2004nonstationary} where the matrix $\mat A$ varies smoothly across space. In that case, we can index the coregionalization matrix $\mat A_i$ with location numbers $i=1,2,\dots,n$. Those would appear in sequence in place of the singular $\mat A$ on the diagonal in equation \eqref{blockCov}. This implies that the inverse and determinant of the covariance matrix implied by the nonstationary version of the LMC can also be computed from those of the individual matrices; this time the $n$ $p\times p$ coregionalization matrices $\mat A_i$ and the $p$ $n\times n$ spatial correlation matrices $\mat R_j$.

 \section{Identifiability} \label{identif}

In the case where exponential correlation functions $\rho_j(\mat s, \mat t) = \exp(-\varphi_j|\mat s - \mat t|)$ are assigned to the base processes $w_j(\cdot), j=1,\dots,p$, we will show that one can identify both the full rank coregionalization matrix $\mat A$ along with range parameters $\varphi_1,\dots,\varphi_p$. We will assume without loss of generality that the range parameters are in increasing order such that $\varphi_1<\dots<\varphi_p$. To be clear, we mean that provided there is a sufficient number $n$ of points (to be clarified below), parameters $\mat A,\varphi_{j=1}^p$ and $\mat B,\psi_{j=1}^p$ will lead to the same model if and only if $\varphi_j = \psi_j$ for all $j$ and $\mat A = \mat B $ up to column sign switches. This is unlike the separable case ($\varphi_1=\varphi_2=\dots=\varphi_p$) where the model is a function of the product $\mat A \mat A^\top$. There exists an infinity of equivalent such factorizations. Standard examples include the lower triangular Cholesky decomposition or the positive definite square root. 

To establish this result, we first look at the cross-covariance function of components $i$ and $j$ evaluated at distance $d$:
$$
C_{ij}(d) = \sum_{k=1}^p a_{ik}a_{jk}\exp(-\varphi_k d).
$$
We postulate for a contradiction that another set of parameters $\mat B,\psi_{j=1}^p$ leads to the same model. This leads to equations  
\begin{align}
f_{ij}(d) := \sum_{k=1}^p  a_{ik} a_{jk}\exp(-\varphi_k d) - \sum_{\ell=1}^p b_{i\ell}b_{j\ell}\exp(-\psi_\ell d) = 0, \qquad i,j=1,\dots,p.\label{LagPoly}
\end{align}
According to Laguerre's generalization \citep{laguerre1884quelques} of Descartes' rule of signs, each $f_{ij}$ above either has at most $2p-1$ positive roots or is identically 0 everywhere. We debunk the latter first.

For $f_{ij}$ to be zero everywhere, it needs to be the case that $\varphi_k = \psi_k$ for all $k$. Otherwise, if there is a $\varphi_k$ that is not equal to any of the alternative range parameters, then it means in particular that the coefficients $a_{ik}^2=0$ for each $i=1,\dots,p$. This signifies a column of 0s in $\mat A$ which is not allowed ($\mat A$ is of full rank). So we can assume $\varphi_k = \psi_k, k=1,\dots,p$ since both sets are ordered. 

Consequently, for the $f_{ii}(d)$ to be 0 everywhere we now need the coefficients $ a_{ik}^2 - b_{ik}^2$ to be equal to zero for all $i,k=1,\dots,p$; every entry of $\mat B$ is equal to the corresponding entry in $\mat A$ up to a change of sign. We also need $a_{ik} a_{jk} - b_{ik}b_{jk}=0$ for every $j=1,\dots,p$ so if an entry $a_{ik} = - b_{ik}$, then the whole column needs to be sign flipped. For all $f_{ij}(d)$ to be zero everywhere, the range parameters need to be the same, and $\mat A$ and $\mat B$ are equal up to column sign changes.

However, it may be the case that the $2p-1$ possible roots of the $f_{ij}$ functions are enough to obtain 2 equivalent models at least at the locations observed. However, if there are more than $2p-1$ distinct distances among the points observed, then there wouldn't be enough roots for the 2 models to agree on all of those. For example, let's say the multivariate process $\mat v(\mat s)$ is observed at locations contained in the plane $\R^2$. If there are $n$ such points, then we know from the Erd{\"o}s distinct distances problem \citep{erdos1946sets} that there are at the very least $\sqrt{n-3/4} - 1/2$ different pairwise distances. It means that if there is a minimum of $(2p-1/2)^2 + 3/4$ points, then the LMC is guaranteed to be identifiable in this context. Most often, points come in no regular pattern and hence present $n(n-1)/2$ distinct distances (as many as there are pairs of points) and such a number usually exceeds $2p-1$ by a large amount.

\section{Cross-Covariance Functions}\label{ccfuncalt}

There is a vast literature on cross-covariance functions for multivariate Gaussian processes. We review some of the modeling advantages and drawbacks, specifically when compared to the LMC. 

In the approach of \citet{apanasovich2010cross} based on latent dimensions, a $p$-variate cross-covariance function is defined from a single univariate one over $\R^{d+k}$ ($1 \leq k \leq p$) by
\begin{align}
C_{ij}(\mat s, \mat t) = C\{(\mat s, \mat \xi_i),(\mat t, \mat \xi_j)\}, \qquad \text{with } \mat s,\mat t \in \R^d; \mat \xi_i,\mat \xi_j \in \R^k. \label{latDim}
\end{align}
The cross-covariance above is valid by construction. The authors use this idea to generalize the class of space-time covariance functions described in \citet{gneiting2002nonseparable} to multivariate data. Ignoring the time component for simplicity, the framework described by equation \eqref{latDim} is employed to construct a cross-covariance function of the form
\begin{align}
C_{ij}(|\mat s - \mat t|) = \frac{\sigma^2}{|\mat \xi_i - \mat \xi_j| + 1} \exp\left(\frac{-\alpha |\mat s - \mat t|}{(|\mat \xi_i - \mat \xi_j| + 1)^{\beta/2}}\right). \label{sepParam}
\end{align}
The interesting part of expression \eqref{sepParam} is the presence of a single parameter $\beta \in [0,1]$ controlling separability. Another notable fact is how parsimonious model \eqref{sepParam} is in terms of parameters. It involves $p$ locations $\mat \xi_1,\dots,\mat \xi_p$, each one in $\R^k$, along with the parameters of the single covariance function over $\R^{d+k}$. In comparison, the LMC has $p^2$ components in its coregionalization matrix $\mat A$ plus the individual parameters contained in the $p$ distinct univariate correlation functions $\rho_j(\cdot)$. We propose an approach to curb this potential over-parametrization problem in Section \ref{sparseApp}.

An inconvenience of using \eqref{latDim}, or the specific instance \eqref{sepParam} of it, is that it leads to identical marginal models for each of the $p$ processes. Perhaps less apparent is the fact that the cross-covariance (for $i \neq j$) functions under \eqref{latDim} are themselves valid univariate covariance functions. For general $d$ and stationary-isotropic covariances that are a function of distance $|\mat s - \mat t|$, that restricts cross-covariances to the class of completely monotone functions (see \citet[Section 2.5.1]{cressie1993statistics}). In particular, negative correlation across processes is prohibited at any distance. \citep{apanasovich2010cross} propose to circumvent those limitations by using linear combinations of multiple processes, something that is very similar in spirit to what the LMC does. 

In terms of univariate correlation functions, there are perhaps none more celebrated than those of the Mat\'ern class \citep{matern1960spatial} defined by
\begin{align}
M(\mat s-\mat t| \nu,\varphi) = \frac{2^{1-\nu}}{\Gamma(\nu)} (\varphi |\mat s-\mat t|)^\nu K_\nu(\varphi |\mat s-\mat t|), \label{maternCorr}
\end{align}
where $K_\nu$ is the modified Bessel function and $\varphi$ is a range parameter. Smoothness is handled by $\nu$ in the sense that any stationary random field with correlation function \eqref{maternCorr} has $\lceil \nu \rceil - 1$ differentiable sample paths in the mean-square sense \citep{guttorp2006studies}. The Mat\'ern class includes the exponential correlation function, among others, as a special case when $\nu=1/2$. 

The objective of a multivariate Mat\'ern covariance function is to have Mat\'ern marginal covariances for the $p$ processes while still allowing cross-dependence among them. \citet{gneiting2010matern} specified a stationary cross-covariance function of the form,
\begin{align}
C_{ij}(|\mat s - \mat t|) = \sigma_{ij} M(\mat s-\mat t| \nu_{ij},\varphi_{ij}), \qquad i,j=1,\dots,p.\label{maternCrossCov}
\end{align}
Conditions need to be imposed on parameters $\nu_{ij},\varphi_{ij}$ to ensure a positive definite covariance. The authors first propose a parsimonious model obtained by applying the covariance convolution approach of \citet{majumdar2007multivariate} (which is itself a rich conceptual framework for constructing valid multivariate covariance functions). However, this approach imposes strong restrictions on the parameters of \eqref{maternCrossCov}. It requires every range parameter $\varphi_{ij}$ to have the same value but still allows each process to have its own differentiability parameter $\nu_{ii} > 0$.

 \citet{gneiting2010matern} also described a more general formulation for $p=2$ by establishing a sufficient condition on $\sigma_{ij}$ under which \eqref{maternCrossCov} constitutes a valid model. This method was extended to higher dimensions in \citet{apanasovich2012valid} although the conditions therein are a bit more involved. The authors do provide some helpful specific case examples. 

A multivariate Mat\'ern model has a specific smoothness parameter for each of the $p$ processes. For the LMC, each component is defined as the linear combination $v_i(\mat s) = \sum_{j=1}^p  a_{ij} w_j(\mat s)$ and therefore will have sample paths that are as differentiable as the roughest of underlying processes $w_j(\cdot)$ (unless structural zeros are imposed on $\mat A$ such as in the triangular case). However, multivariate cross-covariance functions of the form described in \eqref{maternCrossCov} for $i \neq j$ are restricted to monotone functions, so either $C_{ij}$ is positive everywhere or it is negative everywhere. Cross correlation functions implied by the LMC can be positive at close range and negative elsewhere.

Finally, we highlighted in Section \ref{genLMC} how the structure of the LMC can be exploited to allow evaluation of the likelihood in linear computational complexity when it comes to the number of components $p$. For the alternatives we discussed, computations still scale with $p^3$ and therefore quickly become prohibitive for high dimensional processes.

\section{Proof of Proposition \ref{LMCcondDist}} \label{proofCond}

Proposition \ref{LMCcondDist} can be verified by direct and standard calculations. Let $\mat R_j$ be defined as in equation \eqref{fullCorr}, the covariance matrix of the vector $(\mat v_\old ^\top,\mat v_\new ^\top)^\top$ (with $\mat v_\old = \vec(\mat V_\old)$ and $\mat v_\new = \vec(\mat V_\new)$) is of the form $\sum_{j=1}^p \mat R_j \otimes \mat a_j \mat a_j^\top$. Standard conditional results concerning multivariate Gaussians can be applied. We start by first deriving the conditional mean:
$$
E[\mat v_\new | \mat v_\old] = \left( \sum_{i=1}^p \mat R_{\new,\old}^{(i)} \otimes \mat a_i \mat a_i^\top\right)
\left(\sum_{j=1}^p \mat R^{(j)^{-1}}_\old \otimes \mat a_j^{- \top} \mat a_j^{-1}\right)\mat v_\old.
$$
We apply the mixed product property to simplify the first matrix multiplication to 
$$
\sum_{j=1}^p \mat R_{\new,\old}^{(j)} \mat R^{(j)^{-1}}_\old \otimes \mat a_j \mat a_j^{-1}.
$$
Note that the double sum vanishes because $\mat a_i^\top \mat a_j^{- \top}$ is equal to zero whenever $i\neq j$ and 1 otherwise. The conditional mean, in matrix form, is
$$
\sum_{j=1}^p \mat a_j \mat a_j^{-1} \mat V_\old \mat R^{(j)^{-1}}_\old \mat R_{\old,\new}^{(j)}.
$$

For the conditional variance, we have 
\begin{align*}
\var(\mat v_\new | \mat v_\old) = &\sum_{j=1}^p \mat R_\new^{(j)} \otimes \mat a_j \mat a_j^\top 
\\ & \qquad - \left( \sum_{j=1}^p \mat R_{\new,\old}^{(j)} \otimes \mat a_j \mat a_j^\top\right)
\left(\sum_{j=1}^p \mat R^{(j)^{-1}}_\old \otimes \mat a_j^{- \top} \mat a_j^{-1}\right)\left( \sum_{j=1}^p \mat R_{\old,\new}^{(j)} \otimes \mat a_j \mat a_j^\top\right).
\end{align*}
Using the same arguments as for the conditional mean, we can simplify to
$$
\var(\mat v_\new | \mat v_\old) = \sum_{j=1}^p \{\mat R_\new^{(j)} - \mat R_{\new,\old}^{(j)}\mat R^{(j)^{-1}}_\old\mat R_{\old,\new}^{(j)}\}    \otimes \mat a_j \mat a_j^\top .
$$
The LMC covariance structure is preserved under conditioning.

\section{Interweaving}\label{interSec}

Components of the model are $\mat V,\mat W,\mat A, \varphi_{j=1}^p, \tau_{j=1}^p$. Model parameters are $\mat A, \varphi_{j=1}^p, \tau_{j=1}^p$ while $\mat V,\mat W$ are equivalent representations of the latent LMC linked by the 1 to 1 transformation $\mat V=\mat A\mat W$. In the example we present below, we put independent Gaussian priors on the elements of the coregionalization matrix $\mat A$. We use the exponential correlation function $\rho(\varphi_j|\mat s-\mat t|) = \exp(-\varphi_j|\mat s-\mat t|)$ for the underlying process $w_j(\cdot)$ and the range parameters $\varphi_1,\dots,\varphi_p$ are assigned uniform priors on $[3,30]$. Roughly speaking, this corresponds to a practical range at distances between 0.1 and 1 (points are observed in the unit square). Finally, we assign independent inverse Gamma priors with high variance to the parameters $\tau_1,\dots,\tau_p$ to obtain conjugate full conditional updates. 

The global MCMC sampling algorithm we employ is described below. It interweaves updates for $\mat A$ where $\mat V$ and $\mat W$ are conditioned on in turn. 
\begin{algo}[Interweaving MCMC for $\mat A$] {\color{white} help me}
\label{interMCMC}
\begin{enumerate}
    \item Update $\mat V,\mat W$ conditional on $\mat A, \varphi_{j=1}^p, \tau_{j=1}^p$ such that the transition leaves the full conditional invariant:
    \begin{itemize}
        \item Update $\mat V$ (with $\mat A, \varphi_{j=1}^p, \tau_{j=1}^p$ fixed) from full conditionals as detailed in Appendix \ref{datAugStep}, 
        \item Update $\mat W$ conditional on $\mat V, \mat A, \varphi_{j=1}^p, \tau_{j=1}^p$ using the deterministic relation $\mat W = \mat A^{-1}\mat V$.
    \end{itemize}
    \item Update $\mat A,\mat W$ conditional on $\mat V, \varphi_{j=1}^p, \tau_{j=1}^p$ such that the transition leaves the full conditional invariant:
    \begin{itemize}
        \item Update $\mat A$ (with $\mat V, \varphi_{j=1}^p, \tau_{j=1}^p$ fixed) using a slice sampling step \citep{neal2003slice,murray2010slice}, 
        \item Update $\mat W$ conditional on $\mat V, \mat A, \varphi_{j=1}^p, \tau_{j=1}^p$ using the deterministic relation $\mat W = \mat A^{-1}\mat V$.
    \end{itemize}
    \item Update $\mat A,\mat V$ conditional on $\mat W, \varphi_{j=1}^p, \tau_{j=1}^p$ from its exact distribution:
    \begin{itemize}
        \item Update $\mat A$ (with $\mat W, \varphi_{j=1}^p, \tau_{j=1}^p$ fixed) from its exact $p^2$-dimensional Gaussian distribution,
        \item Update $\mat V$ conditional on $\mat W, \mat A, \varphi_{j=1}^p, \tau_{j=1}^p$ using the deterministic relation $\mat V = \mat A\mat W$.
    \end{itemize}
    \item Update $\varphi_{j=1}^p$ conditional on $\mat V, \mat W, \mat A, \tau_{j=1}^p$ using a Metropolis-Hastings move centered at the current values. 
    \item Update $\tau_{j=1}^p$ conditional on $\mat V, \mat W, \mat A, \varphi_{j=1}^p$ from their conjugate inverse gamma distributions.
\end{enumerate}
\end{algo}

In Algorithm \ref{interMCMC}, components $\mat V, \mat W$ and $\mat A$ are sampled at multiple steps of the algorithm. This is referred to as marginalization in collapsed Gibbs samplers theory \citep{van2008partially}. We could omit either of steps 2 or 3 and still have an MCMC algorithm that leaves the joint posterior invariant. If we remove step 2, $\mat A$ is only updated conditional on $\mat W$ and if we remove step 3, $\mat A$ is only updated conditional on $\mat V$. Hence, the two situations correspond respectively to the whitened and centered parametrizations.

We are interested in comparing the interweaving approach with the more standard whitened ($\mat W$ parametrization) and centered ($\mat V$ parametrization) updates in terms of autocorrelation in the posterior samples relating to the matrix $\mat A$. As detailed in Section \ref{genLMC}, the elements of $\mat A$ are prone to both label and sign switching. We instead look at identifiable quantities: the cross-covariances $C_{12}(0) = a_{11} a_{21} + a_{12} a_{22}$ and $C_{12}(0.1) = a_{11} a_{21} \exp(-0.1\varphi_1) + a_{12} a_{22}\exp(-0.1\varphi_2)$ at distances 0 and 0.1 respectively. The first is a function of $\mat A$ while the second also depends on range parameters $\varphi_1$ and $\varphi_2$.

We chose a set of $n=200$ locations randomly and uniformly on the unit square $[0,1]^2$. Importantly, this is done only once and before performing any simulation to not introduce extra variability in the study. Then for each method and StN ratio value, the following is repeated 100 times. We instantiate a $p=2$ dimensional random field at said locations. We use $\mat A = [[-1.5,1.0]^\top,[1.1,2.0]^\top]$ and $[\varphi_1,\varphi_2] = [5,20]$. For the case StN $=1$, we add independent observational noise of variance $\tau_1 = \tau_2 = 1$ at every location; for StN $=0.1$ and StN $=10$ we multiply the noise variables $\mat \epsilon (\mat s)$ by $\sqrt{10}$ and $\sqrt{0.1}$ respectively. We run Algorithm \ref{interMCMC} (omitting step 2 for whitened updates and step 3 for centered updates) for 2000 iterations, discarding the first 1000 as burn-in. Finally, we evaluate the auto-correlation in the Markov chain of both quantities of interest $C_{12}(0)$ and $C_{12}(0.1)$ using the effective sample size metric. This is computed using the CODA package in R \citep{plummer2006coda}. Results are presented in Figure \ref{effSamp}.

\begin{figure}[h]
     \centering
     
        \includegraphics[width=0.8\textwidth]{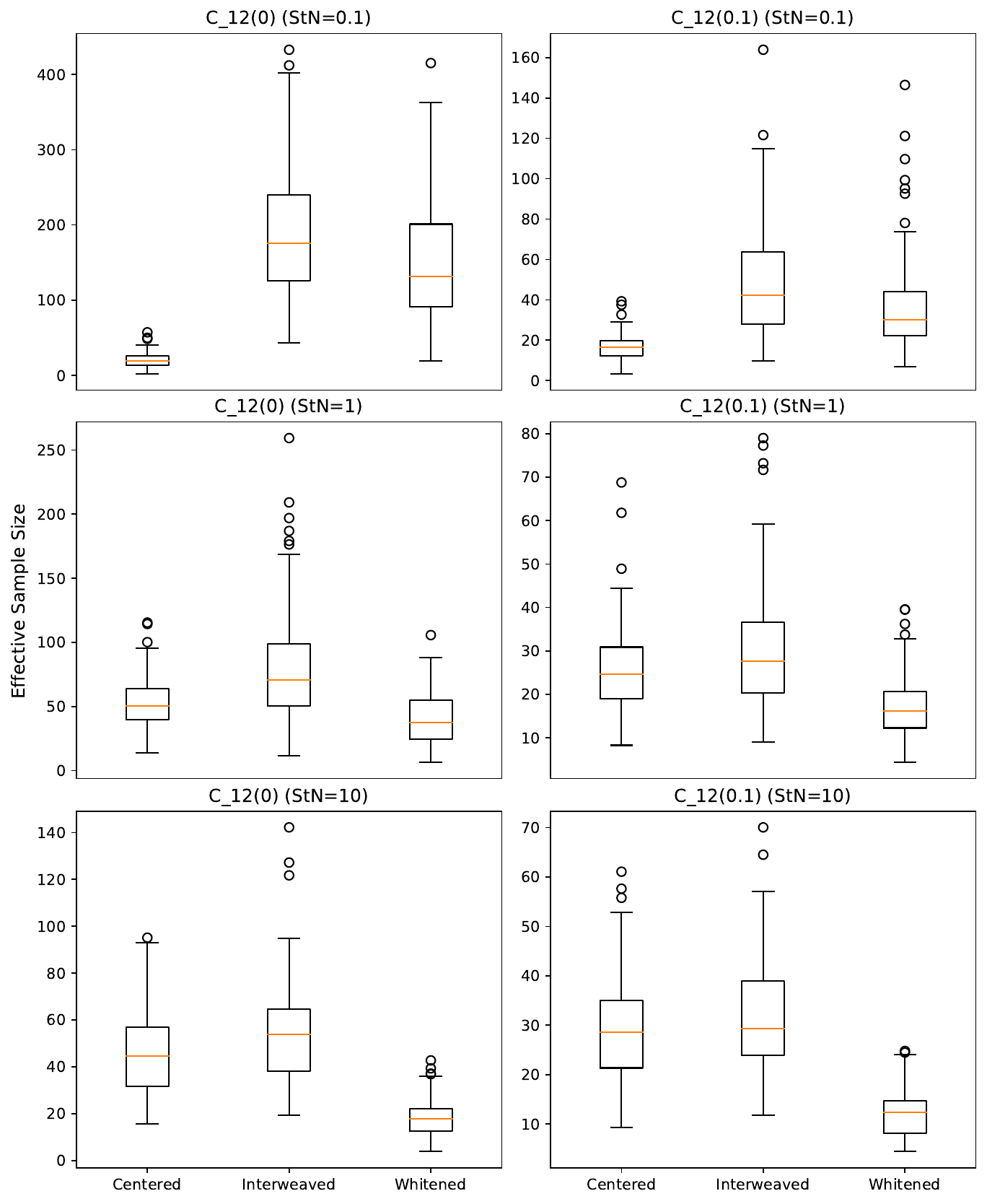}
        \caption{Results of the simulation study described in Appendix \ref{interSec}. Each boxplot contains 100 samples obtained by conducting inference on a new set of observations each time. The effective sample size is calculated for the two Markov chains of length $1000$ corresponding to quantities of interest $C_{12}(0)$ and $C_{12}(0.1)$. The average computing time for the centered, interweaved and whitened algorithms were respectively 33.46, 33.60 and 31.77 seconds.}
        \label{effSamp}
\end{figure}

As expected, the whitened (resp. centered) approach performs poorly when the StN ratio is low (resp. high). The interweaving approach seems to always perform at least slightly better than the most suitable of the other two methods. In practice, we do not know in advance how process level and observational variances compare. Interweaving, as was the objective in its conceptualization \citep{yu2011center}, can benefit from the advantages of both worlds, eliminating the need to decide upon an appropriate parametrization. 

It is important to emphasize that interweaving samples come at no added computational cost. For all practical purposes, computing time is identical across the three methods (about 30 seconds per chain of length 2000). This is because the computation is dominated by steps 1 and 4 of Algorithm \ref{interMCMC} which require respectively $\cO(p^2n^2)$ (see Appendix \ref{datAugStep}) and $\cO(pn^3)$ floating point operations (flops) (inverse and determinant of the $n \times n$ matrices $\mat R_j,j=1,\dots,p$). In the vectorized interpretation of the LMC, updating the matrix $\mat A$ in the centered parametrization would involve computing the inverse and determinant of the $np \times np$ covariance matrix $\sum_{j=1}^p \mat R_j \otimes \mat a_j \mat a_j^\top$. 

A consequence of Proposition \ref{LMCdensRes} is that computations on the coregionalization matrix $\mat A$ and the spatial correlation matrices $\mat R_1,\dots,\mat R_p$ are disentangled when evaluating the likelihood. Updates to $\mat A$ are negligible when compared to updates of spatial ranges $\varphi_1, \dots, \varphi_p$. Nevertheless, there are $p^2$ coregionalization parameters compared to $p$ spatial parameters, hence being able to use more sophisticated samplers on $\mat A$ (at no additional cost) improves the convergence rate of the MCMC algorithm. 

\section{Instantiating the Latent GP} \label{datAugStep}
 In Section \ref{hierModels}, we described the common issue regarding computational scalability in latent Gaussian process models. We first do a high-level recap here. Suppose you have an additive Gaussian model of the form $\mat y = \mat v + \mat \epsilon$ where $\mat y$ is observed, $\mat v$ has a multivariate normal distribution accounting for dependence among observations (such as spatial correlation accounting for proximity) and $\mat \epsilon$ represents observational noise that is independent across observations.

 We assume $\mat v$ to have the $n$-dimensional normal distribution with mean 0 and covariance matrix $\mat \Sigma$. In general, evaluating the density of $\mat v$ requires $\cO(n^3)$ flops to compute both the determinant and inverse of $\mat \Sigma$. We might be tempted to use an approximate model instead or impose some structure on $\mat \Sigma$ such that both $\mat \Sigma^{-1}$ and $\det(\mat \Sigma)$ are easier to compute. The issue is that the marginal model for $\mat y$ has covariance matrix $\mat \Sigma + \mat D$ where $\mat D$ is the diagonal covariance matrix associated with $\mat \epsilon$. Even if we employ a scalable model for $\mat \Sigma$, there is no simple way to extend such a strategy for computing the inverse and determinant of the covariance matrix of $\mat y$. In other words, there is no straightforward link between the easier computed $\mat \Sigma^{-1}$ and $\det(\mat \Sigma)$ and the analogous for $\mat \Sigma + \mat D$ because $\mat D$ is of full rank.

Even though $\mat v$ cannot easily be marginalized, there is always the option to rely on data augmentation and instantiate this latent vector at every step of the Markov chain. In this case, we work with the complete data likelihood given by
\begin{align}
p(\mat y,\mat v) = \frac{\exp(-\frac12 [\mat v^\top \mat \Sigma^{-1} \mat v + (\mat y - \mat v)^\top\mat D^{-1}(\mat y - \mat v)])}{(2\pi)^{n/2} \det(\mat \Sigma)^{1/2}\det(\mat D)^{1/2}} \label{jointNorm}
\end{align}
which is more amenable to likelihood-based inference since it only involves the inverse and determinant of the well-behaved $\mat \Sigma$ matrix and the diagonal $\mat D$.

We are left with only the matter of instantiating $\mat v$ conditional on $\mat y$. From standard manipulations of equation \eqref{jointNorm}, which is proportional to the conditional, we conclude that $\mat v|\mat y$ is multivariate normal with mean $(\mat \Sigma^{-1}+\mat D^{-1})^{-1}\mat D^{-1}\mat y$ and variance $(\mat \Sigma^{-1}+\mat D^{-1})^{-1}$. We again deal with matrices of size $n \times n$ which we need to invert and, in this case, factorize to simulate this $n$-dimensional Gaussian. In terms of asymptotic complexity, it renders pointless our attempt to use a scalable model for $\mat v$ as we would still need to perform $\cO(n^3)$ computations.

We have the option of updating $\mat v$ using some variant of Metropolis-Hastings. Acceptance ratios are readily computed using \eqref{jointNorm}. However, $\mat v$ is unlikely to move a lot in high-dimensional settings unless careful attention is put into tailoring and tuning the proposals for the specific problem at hand. A more universal approach is to simulate the values of $\mat v$ one at a time from the full conditionals described by
\begin{align}
\var (v_i|\mat v_{-i}, \mat y) &= \frac{1}{\{\mat \Sigma^{-1}\}_{ii} + \mat D_{ii}^{-1}}, \label{fullConVar}\\
E[v_i|\mat v_{-i}, \mat y] &= \var (v_i|\mat v_{-i}, \mat y)\left(\mat D_{ii}^{-1} y_i - \sum_{j\neq i} \{\mat \Sigma^{-1}\}_{ij}v_j\right).\label{fullCondMean}
\end{align}

Variance is directly obtained from $\mat \Sigma^{-1}$ and $\mat D^{-1}$ while the mean is obtained from a computation linear in $n$. Repeat that for every $i=1,\dots,n$ and $\mat v$ is updated in $\cO(n^2)$ flops. 

In the case of the LMC, we can explicitly compute the precision matrix $\sum_{j=1}^p \mat R_j^{-1} \otimes \mat a_j^{-\top} \mat a_j^{-1}$ of the $np$-dimensional vector $\mat v = \vec  (\mat V)$ and then perform the update from full conditionals in $\cO(p^2n^2)$ complexity. However, we do not require computing the inverse of the $np \times np$ covariance matrix to evaluate the likelihood (equation \eqref{LMCdens}). Doing so entails $\cO(p^3n^2)$ operations and therefore would dominate the spatial latent effect update.

We instead update the latent $p \times n$ matrix $\mat V = (\mat v(\mat s_1),\dots,\mat v(\mat s_n))$ column by column using the $p$-dimensional Gaussian full conditionals. From the joint likelihood of $\mat V$ and $\mat Y$ in equation \eqref{compCent}, we can derive
\begin{align}
\var(\mat v(\mat s_i)|\mat V_{-i}) &= \left ( \sum_{j=1}^p   \left [\mat R_{j}^{-1}  \right]_{ii}\mat a_j^{-\top} \mat a_j^{-1} + \mat D^{-1} \right )^{-1},\label{fcLmcVar}\\
E[\mat v(\mat s_i)|\mat V_{-i}] &= \var(\mat v(\mat s_i)|\mat V_{-i}) \left ( \mat D^{-1} \mat y(\mat s_i) - \sum_{k\neq i}\sum_{j=1}^p  \left [\mat R_{j}^{-1}  \right]_{ik}\mat a_j^{-\top} \mat a_j^{-1} \mat v(\mat s_k) \right ), \label{fcLmcMean}
\end{align}
where $\mat D$ is the $p \times p$ diagonal matrix with elements $\tau_i, i=1,\dots,p$. The $p$-dimensional full conditional update described by equations \eqref{fcLmcVar} and \eqref{fcLmcMean} above is accomplished in $\cO(p^2n)$ operations provided we start by computing the inner products $\mat a_j^{-1} \mat v(\mat s_k)$ in the mean structure. Consequently, the global column-by-column update to the matrix $\mat V$ has $\cO(p^2n^2)$ complexity.

A common occurrence in multivariate spatial modeling is the presence of missing entries in the matrix $\mat Y$ of observations. In this case, we can still fully instantiate the matrix $\mat V$ conditional on the available values in $\mat Y$. It necessitates a slight adaptation of the procedure outlined above. From a numerical standpoint, it will be simpler to consider the matrix $\mat Y$ to be full of scalar values. We shall input arbitrary values in place of the missing entries. We employ a $p \times n$ binary masking matrix $\mat M = [\mat m (\mat s_1),\dots,\mat m(\mat s_n) ]$ to encode (with 0s) the position of missing values. In the complete likelihood of $\mat Y,\mat V$ in equation \eqref{compCent}, the influence of $\mat Y$ acts through the sum of quadratic products described by
\begin{align}
\trace\{(\mat V - \mat Y)^\top \mat D^{-1} (\mat V - \mat Y)\} = \sum_{i=1}^n (\mat v(\mat s_i) - \mat y(\mat s_i))^\top \mat D^{-1} (\mat v(\mat s_i) - \mat y(\mat s_i)). \label{quads}
\end{align}
To discount elements that would involve a value $y_j(\mat s_i)$ that is missing for some $j$ in $1,\dots,p$, we can simply replace the $p \times p$ matrix $\mat D^{-1}$ on the RHS of equation \eqref{quads} by $\mat D^{-1} \circ \diag \{\mat m(\mat s_i)\}$ where $\circ$ stands for the element-wise (Hadamard) product of two matrices. It means that we can simply replace $\mat D^{-1}$ by $\mat D^{-1} \circ \diag \{\mat m(\mat s_i)\}$ in equations \eqref{fcLmcVar} and \eqref{fcLmcMean} and otherwise carry out the updates to each vector $\mat v(\mat s_i)$ unchanged.

\section{Algorithms} \label{algosDetails}

In this section, we present the structure of the MCMC algorithms we use for both the regular LMC and our sparse version. Each algorithm can be decomposed in four steps corresponding to the sampling of the latent spatial matrix $\mat V$, covariance parameters $\mat A, \varphi_{j=1}^p$ and the observational variances $\tau_{j=1}^p$. We first present the MCMC procedure for sampling quantities of interest in a zero-mean standard LMC. 

\begin{algo}[Standard MCMC for LMC distributed random fields] {\color{white} help me}
\label{baseLMCalgo}
\begin{enumerate}
    \item Update $\mat V$ conditional on $\mat Y,\mat A, \varphi_{j=1}^p, \tau_{j=1}^p$ from full conditionals as detailed in Appendix \ref{datAugStep}. 

    \item Update $\mat A$ conditional on $\mat Y,\mat V, \varphi_{j=1}^p, \tau_{j=1}^p$ using a slice sampling step \citep{neal2003slice,murray2010slice}.

    \item Update $\varphi_{j=1}^p$ conditional on $\mat Y,\mat V, \mat A, \tau_{j=1}^p$ using a Metropolis-Hastings move centered at the current values. 
    \item Update $\tau_{j=1}^p$ conditional on $\mat Y,\mat V,  \mat A, \varphi_{j=1}^p$ from their conjugate inverse gamma distributions.
\end{enumerate}
\end{algo}

For the sparse LMC, we need to be able to perform updates that explore admissible models. We consider two types of reversible jump proposals: one that removes a structural zero and one that inserts a structural zero. The latter move needs to be valid in the sense that it does not render the coregionalization matrix $\mat A$ singular with probability one (see Section \ref{sparseApp}).

\begin{algo}[MCMC for the sparse LMC model] {\color{white} help me}
\label{sparseLMCalgo}
\begin{enumerate}
    \item Update $\mat V$ conditional on $\mat Y,\mat A, \varphi_{j=1}^p, \tau_{j=1}^p$ from full conditionals as detailed in Appendix \ref{datAugStep}. 

    \item Update $\mat A$ conditional on $\mat Y, \mat V, \varphi_{j=1}^p, \tau_{j=1}^p$:
    \begin{itemize}
        \item Perform $p$ updates to the structure of $\mat A$ (with $\mat V, \varphi_{j=1}^p, \tau_{j=1}^p$ fixed) using reversible jump proposals, 
        \item Update the non-zero values of $\mat A$ (with $\mat V, \varphi_{j=1}^p, \tau_{j=1}^p$ fixed) using a slice sampling step \citep{neal2003slice,murray2010slice}.
    \end{itemize}

    \item Update $\varphi_{j=1}^p$ conditional on $\mat Y,\mat V, \mat A, \tau_{j=1}^p$ using a Metropolis-Hastings move centered at the current values. 
    \item Update $\tau_{j=1}^p$ conditional on $\mat Y, \mat V, \mat A, \varphi_{j=1}^p$ from their conjugate inverse gamma distributions.
\end{enumerate}
\end{algo}

Under the complete data likelihood of $\mat Y,\mat V$ described by equation \eqref{compCent}, both algorithms run at about the same speed. This is assuming the number $n$ of observations is considerably larger than the number $p$ of components. First, let us consider the form of the likelihood ratio upon an update to the range parameter $\varphi_{j=1}^p$. The move proposes to replace each $n \times n$ spatial correlation matrix $\mat R_j$ with a new $\mat R'_j$ for $j=1,\dots,p$. The computational bottleneck of this update is the following likelihood ratio which is necessary for determining the accept rule:
\begin{align}
    \frac{\cL(\mat R_1,\dots,\mat R_p;\mat V)}{\cL(\mat R'_1,\dots,\mat R'_p;\mat V)} = 
\exp \left[ -\frac{1}{2} \sum_{j=1}^p \mat a_j^{-1}\mat V  (\mat R_j^{-1}-{\mat R'_j}^{-1}) \mat V^\top \mat a_j^{-\top}\right] \prod_{j=1}^p \left[\frac{\det(\mat R'_j)}{\det(\mat R_j)}\right]^{\frac 12}.\label{lrR}
\end{align}
The update requires to invert and compute the determinant of the $p$ $n \times n$ correlation matrices 
$\mat R'_j,j=1,\dots,p$ for a calculation that scales as $\cO(pn^3)$. On the other hand, an update to $\mat A$ requires the likelihood ratio described by 
\begin{align}
\frac{\cL(\mat A;\mat V)}{\cL(\mat A';\mat V)} = \etr \left[ -\frac{1}{2} \sum_{j=1}^p   (\mat a_{j}^{-\top}\mat a_{j}^{-1}-{\mat a'_{j}}^{-\top}{\mat a'_{j}}^{-1})\mat V \mat R_j^{-1} \mat V^\top \right] \left |\frac{\det(\mat A')}{\det(\mat A)} \right | ^n,\label{lrA}
\end{align}
where $\etr(\cdot) := \exp(\trace(\cdot))$. When we calculate the quantity above, the inverse matrices $\mat R_j^{-1}$ are already computed and stored. Empirically, the fact that we perform multiple updates to the matrix $\mat A$, both in the slice sampling step and the reversible jumps proposals, has little effect on the global computation. This is because the time required to compute the likelihood ratio upon an update to $\mat A$ (equation \eqref{lrA}) is negligible when compared with the likelihood ratio in \eqref{lrR}. If we were to consider the $np$-dimensional Gaussian distribution of the vector $\vec(\mat V)$ as is typically the case in the multivariate spatial literature, then an update to $\mat A$ or the range parameters $\varphi_{j=1}^p$ would entail a change in the $np \times np$ covariance matrix $\sum_{j=1}^p \mat R_j \otimes \mat a_j \mat a_j^\top$ and its inverse and determinant would need to be recomputed, which is very expensive. 

\section{MCMC Setup} \label{MCMCsetup}

In this section, we describe in more detail the prior distributions employed in the analysis of Sections \ref{simStud} and \ref{anCaliData}. We also provide the likelihood functions used and the target posterior distribution in each case. 

\subsection*{Section \ref{simStud}: Simulation Study on the Sparse LMC}
\subsubsection*{Prior Distributions}

\begin{itemize}
    \item Elements of the coregionalization matrix $\mat A$ are assigned independent normal priors with mean 0 and variance equal to 1: $p(\mat A) \propto \prod_{j=1}^p\prod_{i=1}^p \exp(-a_{ij}^2/2)$.
    \item The range parameters $\varphi_{j=1}^p$ are assigned independent uniform priors on the [3,30] interval: $p(\varphi_{j=1}^p) \propto \prod_{j=1}^p \ind_{[3,30]}(\varphi_j)$.
    \item The observational variances $\tau_{j=1}^p$ are assigned independent inverse gamma prior distributions with shape and scale parameters equal to 1: $p(\tau_{j=1}^p) \propto \prod_{j=1}^p \tau_j^{-2} \exp(-1/\tau_j) $.
    \item For the masking matrix $\mat M$ (sparse LMC), we employ the prior distribution (we use $\pi = 1/p$) described in Section \ref{sparseApp}: $p(\mat M) \propto (1/p)^{\sum m_{ij}} (1-1/p)^{p^2 - \sum m_{ij}}$. 
\end{itemize}

\subsubsection*{Likelihood}
We make use of the complete data likelihood consisting of the matrix $\mat Y$ of observations and the matrix $\mat V$ of latent spatial effects:
\begin{align*}
p(\mat Y, \mat V | \mat A, \mat R_{j=1}^p, \mat D) = \frac{\exp[-\frac{1}{2} \sum_{j=1}^p \mat a_j^{-1}\mat V \mat R_j^{-1} \mat V^\top \mat a_j^{-\top} -\frac{1}{2} \trace\{(\mat V - \mat Y)^\top \mat D^{-1} (\mat V - \mat Y)\}]}{ (2\pi)^{np} |\det(\mat A)|^{n} \prod_{j=1}^p \det(\mat R_j)^{1/2}\prod_{j=1}^p \tau_j^{n/2}}. 
\end{align*}
The vectors $\mat a_j^{-\top},j=1,\dots,p$ are the rows of the inverse matrix $\mat A^{-1}$. The matrices $\mat R_j,j=1,\dots,p$ are the correlation matrices computed respectively from each range parameter $\varphi_{j=1}^p$ using the exponential correlation function. Finally, the matrix $\mat D$ is $p \times p$ diagonal with each term corresponding to the observational variances $\tau_{j=1}^p$. In the case of the sparse LMC, we replace the coregionalization matrix $\mat A$ by the element-wise product $\mat A \circ \mat M$.

\subsubsection*{Posterior Distribution}

The target posterior distribution is proportional to the product of priors with the complete data density:
$$
p(\mat V, \mat A, \varphi_{j=1}^p, \tau_{j=1}^p| \mat Y) \propto p(\mat Y, \mat V | \mat A, \mat R_{j=1}^p, \mat D) p(\mat A) p(\varphi_{j=1}^p) p(\tau_{j=1}^p).
$$
For the sparse LMC, we also sample among admissible models described by the masking matrix $\mat M$:
$$
p(\mat V, \mat A, \mat M, \varphi_{j=1}^p, \tau_{j=1}^p| \mat Y) \propto p(\mat Y, \mat V | \mat A \circ \mat M, \mat R_{j=1}^p, \mat D) p(\mat A) p(\varphi_{j=1}^p) p(\tau_{j=1}^p)p(\mat M).
$$

\subsection*{Section \ref{anCaliData}: Analysis of the California Air Data}
\subsubsection*{Prior Distributions}

\begin{itemize}
    \item Elements of the coregionalization matrix $\mat A$ are assigned independent normal priors with mean 0 and variance equal to 1: $p(\mat A) \propto \prod_{j=1}^p\prod_{i=1}^p \exp(-a_{ij}^2/2)$.
    \item The range parameters $\varphi_{j=1}^p$ are assigned independent uniform priors on the [3,30] interval: $p(\varphi_{j=1}^p) \propto \prod_{j=1}^p \ind_{[3,30]}(\varphi_j)$.
    \item The observational variances $\tau_{j=1}^p$ are assigned independent inverse gamma prior distributions with shape and scale parameters equal to 1: $p(\tau_{j=1}^p) \propto \prod_{j=1}^p \tau_j^{-2} \exp(-1/\tau_j) $.
    \item Values of component-specific mean vector $\mat \mu$ are each assigned independent normal priors with mean 0 and variance equal to 10: $p(\mat \mu) = \prod_{j=1}^p \exp(-\mu_{j}^2/20) $.
    \item For the masking matrix $\mat M$, we employ the prior distribution (we use $\pi = 1/p$) described in Section \ref{sparseApp}: $p(\mat M) \propto (1/p)^{\sum m_{ij}} (1-1/p)^{p^2 - \sum m_{ij}}$. 
\end{itemize}

\subsubsection*{Likelihood}
We make use of the complete data likelihood consisting of the matrix $\mat Y$ of observations and the matrix $\mat V$ of latent spatial effects:
\begin{align*}
p(\mat Y, \mat V |\mat \mu, \mat A, \mat R_{j=1}^p, \mat D) = \frac{\exp[-\frac{1}{2} \sum_{j=1}^p \mat a_j^{-1}(\mat V - \mat \mu\mat 1_n^\top) \mat R_j^{-1} (\mat V - \mat \mu\mat 1_n^\top)^\top \mat a_j^{-\top} -\frac{1}{2} \trace\{(\mat V - \mat Y)^\top \mat D^{-1} (\mat V - \mat Y)\}]}{ (2\pi)^{np} |\det(\mat A)|^{n} \prod_{j=1}^p \det(\mat R_j)^{1/2}\prod_{j=1}^p \tau_j^{n/2}}. 
\end{align*}
The vectors $\mat a_j^{-\top},j=1,\dots,p$ are the rows of the inverse matrix $\mat A^{-1}$. The matrices $\mat R_j,j=1,\dots,p$ are the correlation matrices computed respectively from each range parameter $\varphi_{j=1}^p$ using the exponential correlation function. Finally, the matrix $\mat D$ is $p \times p$ diagonal with each term corresponding to the observational variances $\tau_{j=1}^p$. We replace the coregionalization matrix $\mat A$ by the element-wise product $\mat A \circ \mat M$. The $j^{\text {th}}$ row of $p \times n$ matrix $\mat \mu\mat 1_n^\top$ is filled with the value $\mu_j,j=1,\dots,p$.

\subsubsection*{Posterior Distribution}

The target posterior distribution is proportional to the product of priors with the complete data density:
$$
p(\mat V,\mat \mu, \mat A, \mat M, \varphi_{j=1}^p, \tau_{j=1}^p| \mat Y) \propto p(\mat Y, \mat V | \mat A \circ \mat M, \mat R_{j=1}^p, \mat D) p(\mat A) p(\varphi_{j=1}^p) p(\tau_{j=1}^p) p(\mat \mu)p(\mat M).
$$

\end{document}